\documentclass[11pt]{article}
\usepackage{amsmath,amsfonts,amssymb,amsthm}
\usepackage{fullpage}
\usepackage[ruled]{algorithm2e}

\usepackage{subfigure} \usepackage{epstopdf}
\usepackage{graphicx}
\usepackage{color}
\usepackage[sans]{dsfont}
\usepackage{mathtools}
\usepackage{hyperref}
\usepackage{cleveref}
\usepackage{tcolorbox}
\usepackage{multirow}
\usepackage[shortlabels]{enumitem}
\usepackage{xspace}
\tcbuselibrary{skins,breakable}
\tcbset{enhanced jigsaw}
\usepackage[font=small,labelfont=bf]{caption}

\newtheorem{theorem}{Theorem}

\newtheorem{lemma}[theorem]{Lemma}
\newtheorem{observation}[theorem]{Observation}

\theoremstyle{definition}
\newtheorem{claim}[theorem]{Claim}

\newcommand{\set}[1]{\left\{ #1 \right\}}

\newcommand{\uset}{{\mathcal U}}

\newcommand{\dset}{{\mathcal{D}}}

\newcommand{\lset}{{\mathcal{L}}}

\newcommand{\fset}{{\mathcal{F}}}

\newcommand{\gset}{{\mathcal{G}}}

\newcommand{\hset}{{\mathcal{H}}}

\newcommand{\eps}{{\varepsilon}}
\newcommand{\poly}{\mathsf{poly}}

\newcommand{\dist}{\textnormal{\textsf{dist}}}

\newcommand\vol{\mathsf{cost}}

\def\floor#1{\left\lfloor #1 \right\rfloor}
\def\ceil#1{\left\lceil #1 \right\rceil}

\def\card#1{\left| #1 \right|}
\def\deg{\textrm{deg}}

\def\pr#1{\mathrm{Pr}\left[ #1 \right]}

\def\ex#1{{\mathbb{E}}\left[ #1 \right]}

\newcounter{note}

\newcommand{\mc}{\mathsf{mincut}}
\newcommand{\pat}{\Phi}

\begin{document}

\begin{titlepage}
	
	\title{Cut-Preserving Vertex Sparsifiers for Planar and Quasi-bipartite Graphs}
	
	\author{Yu Chen\thanks{EPFL, Lausanne, Switzerland. Email: {\tt yu.chen@epfl.ch}. Supported by ERC Starting Grant 759471.} \and Zihan Tan\thanks{Rutgers University, NJ, USA. Email: {\tt zihantan1993@gmail.com}. Supported by a grant to DIMACS from the Simons Foundation (820931).}} 
	
	\maketitle

	\thispagestyle{empty}
	\begin{abstract}
We study vertex sparsification for preserving cuts.
Given a graph $G$ with a subset $|T|=k$ of its vertices called terminals, a \emph{quality-$q$ cut sparsifier} is a graph $G'$ that contains $T$, such that, for any partition $(T_1,T_2)$ of $T$ into non-empty subsets, the value of the min-cut in $G'$ separating $T_1$ from $T_2$ is within factor $q$ from the value of the min-cut in $G$ separating $T_1$ from $T_2$. 
The construction of cut sparsifiers with good (small) quality and size has been a central problem in graph compression for years.

Planar graphs and quasi-bipartite graphs are two important special families studied in this research direction.
The main results in this paper are new cut sparsifier constructions for them in the high-quality regime (where $q=1$ or $1+\eps$ for small $\eps>0$).


We first show that every planar graph admits a planar quality-$(1+\eps)$ cut sparsifier of size $\tilde O(k/\poly(\eps))$, which is in sharp contrast with the lower bound of $2^{\Omega(k)}$ for the quality-$1$ case.

We then show that every quasi-bipartite graph admits a quality-$1$ cut sparsifier of size $2^{\tilde O(k^2)}$. This is the second to improve over the doubly-exponential bound for general graphs (previously only planar graphs have been shown to have single-exponential size quality-$1$ cut sparsifiers).

Lastly, we show that contraction, a common approach for constructing cut sparsifiers adopted in most previous works, does not always give optimal bounds for cut sparsifiers. We demonstrate this by showing that the optimal size bound for quality-$(1+\eps)$ contraction-based cut sparsifiers for quasi-bipartite graphs lies in the range $[k^{\tilde\Omega(1/\eps)},k^{O(1/\eps^2)}]$, while in previous work an upper bound of $\tilde O(k/\eps^2)$ was achieved via a non-contraction approach.

	\end{abstract}
\end{titlepage}

\renewcommand{\baselinestretch}{0.75}\normalsize
\tableofcontents
\renewcommand{\baselinestretch}{1.0}\normalsize

\newpage

\section{Introduction}

Given a graph $G$ and a set $T$ of its vertices called \emph{terminals}, a \emph{cut sparsifier} of $G$ with respect to $T$ is a graph $G'$ with $T\subseteq V(G')$, such that for every partition $(T_1,T_2)$ of $T$ into non-empty subsets, the size of the minimum cut separating $T_1$ from $T_2$ in $G$ and the size of the minimum cut separating $T_1$ from $T_2$ in $G'$ are within some small multiplicative factor $q\ge 1$, which is also called the \emph{quality} of the sparsifier.

Karger \cite{karger1993global,karger1999random} first considered the special case where $T=V(G)=V(G')$ and $G'$ is required to be a subgraph of $G$, and
he used sampling approaches to construct quality-$(1+\varepsilon)$ cut sparsifiers with $O(n\log n/\varepsilon^2)$ edges.
Such sparsifiers are called \emph{edge sparsifiers}, since we are only allowed to sparsify the edges (not the vertices) and the goal is to minimize $|E(G')|$.
In this paper, we allow $G'$ to have a different set of vertices, and such graphs $G'$ are known as \emph{vertex sparsifiers}.
Constructing cut-preserving vertex sparsifiers with small quality and size (measured by the number of vertices in $G'$) has been a central problem in graph compression.

Cut-preserving vertex sparsifiers were first studied in the special case where $G'$ is only allowed to contain terminals (that is, $V(G')=T$). Moitra \cite{moitra2009approximation} first showed that every  $k$-terminal graph has a cut sparsifier with quality $O(\log k/\log\log k)$.
Then Charikar, Leighton, Li and Moitra \cite{charikar2010vertex} showed that such sparsifiers can be computed efficiently.
The best quality lower bound in this case is $\Omega(\sqrt{\log k}/\log\log k)$ \cite{makarychev2010metric,charikar2010vertex}, leaving a small gap to be closed. 
Therefore, the natural next question, which has been a central question on cut/flow sparsifiers over the past decade, is:

\vspace{-10pt}
\[\emph{Can better quality sparsifiers be achieved by allowing a small number of Steiner vertices}? 
\]
\vspace{-10pt}

For exact cut sparsifiers ($q=1$), it has been shown that every $k$-terminal graph admits an exact cut sparsifier of size at most $2^{2^k}$ \cite{hagerup1998characterizing,khan2014mimicking}, while the strongest lower bound is $2^{\Omega(k)}$ \cite{khan2014mimicking,karpov2017exponential}, leaving an exponential gap yet to be closed.
Since then the research on cut sparsifiers was diverted into two directions.
The first direction focuses on proving better size upper bounds for special families of graphs.
For example, for planar graphs, Krauthgamer and Rika \cite{krauthgamer2013mimicking,krauthgamer2017refined} showed that $2^{O(k)}$ vertices are sufficient for an exact cut sparsifier, and then Karpov, Philipzcuk and Zych-Pawlewicz \cite{karpov2017exponential} proved a lower bound of $2^{\Omega(k)}$, showing that the single exponential upper bound is almost tight.
Another example is graphs where each terminal has degree exactly $1$. For these graphs, Chuzhoy \cite{chuzhoy2012vertex} showed the construction for $O(1)$-quality sparsifiers of size $O(k^3)$ by contracting expanders, and Kratsch and Wahlstrom \cite{kratsch2012representative} constructed quality-$1$ sparsifiers of size $O(k^3)$ via a matroid-based approach.
The second direction focuses on constructing cut sparsifiers with a slightly worse quality of $(1+\eps)$ for small $\eps>0$.
Andoni, Gupta, and Krauthgamer \cite{andoni2014towards} initiated the study of quality-$(1+\eps)$ flow sparsifiers \footnote{Flow sparsifiers are concretely related to cut sparsifiers. We review the previous work and discuss the connection between them in \Cref{sec: related}.} (which are stronger versions of cut sparsifiers), established a framework of constructing flow sparsifiers via contractions with a doubly exponential upper bound. However, the only family that they managed to obtain a better-than-$2^{2^k}$ upper bound for is the family of quasi-bipartite graphs, graphs where every edge is incident to some terminal. They showed an upper bound of $\poly(k/\eps)$, via a sampling-based approach.
For quality-$(1+\eps)$ cut sparsifiers on quasi-bipartite graphs, the bound was improved to $\tilde O(k/\eps^2)$ by Jambulapati, Lee, Liu and Sidford \cite{jambulapati2023sparsifying}, also via a sampling-based approach.

Among the above algorithmic results, the most common approach for constructing cut sparsifiers is contraction. Specifically, we compute a partition $\fset$ of the vertices of $G$ into disjoint subsets, and then contract each subset of $\fset$ into a supernode.
The resulting graph is called a \emph{contraction-based cut sparsifier}.
Contraction-based cut sparsifiers may only increase the value of every terminal min-cut, which partly explains why contraction is a commonly adopted approach, as to analyze its quality it suffices to show that the min-cuts are not increased by too much.
For general graphs, a lower bound of $2^{2^{\Omega(k)}}$ (almost matching the upper bound of $2^{2^k}$) was proved \cite{karpov2017exponential} for contraction-based cut sparsifiers. It is therefore natural to ask whether or not better bounds can be achieved for special graphs.
Moreover, all the previous constructions of cut sparsifiers with Steiner nodes, except for the ones using matroid-based/sampling-based approaches, are contraction-based.
Therefore, it is natural to explore if contraction can give optimal constructions for cut sparsifiers.

\subsection{Our Results}

In this paper, we make progress on the above questions. Our first results, summarized as \Cref{main: upper} and \Cref{quasi_1}, are new size upper bounds on exact or quality-$(1+\eps)$ sparsifiers for special families of graphs.

\begin{theorem}
\label{main: upper}
For every integer $k\ge 1$ and real number $\eps>0$, every planar graph with $k$ terminals admits a quality-$(1+\eps)$ cut sparsifier on $O(k\cdot\poly(\log k/\eps))$ vertices, which is also a planar graph.
\end{theorem}

For planar graphs, our near-linear size bound is in sharp contrast with the single-exponential bound $2^{\Theta(k)}$ \cite{krauthgamer2013mimicking,krauthgamer2017refined,karpov2017exponential} for exact cut sparsifiers. In other words, \Cref{main: upper} shows that, for preserving cuts in planar graphs, we can trade a small loss in quality for a significant improvement in size.

\begin{theorem}
	\label{quasi_1}
For every integer $k\ge 1$, every quasi-bipartite graph with $k$ terminals admits an exact (quality-$1$) contraction-based cut sparsifier of size $2^{O(k^2\log k)}$, and there exists a quasi-bipartite graph with $k$ terminals whose exact contraction-based cut sparsifier must contain $\Omega(2^k)$ vertices.
\end{theorem}

\Cref{quasi_1} shows that, similar to planar graphs, the family of quasi-bipartite graphs also admit a single-exponential upper bound (even achievable by contraction-based sparsifiers), better than the doubly-exponential bound for general graphs \cite{hagerup1998characterizing,khan2014mimicking,karpov2017exponential}.

Next we turn to study contraction-based sparsifiers.
Our next result shows that contraction-based sparsifiers are not optimal constructions, as they give worse bounds for quasi-bipartite graphs.

\begin{theorem}
\label{main: lower}
For every integer $k\ge 1$ and real number $\eps>0$, there exists a quasi-bipartite graph with $k$ terminals, whose quality-$(1+\eps)$ contraction-based cut sparsifier must contain $k^{\tilde\Omega(1/\eps)}$ vertices.
\end{theorem}

Compared with the previous bounds $\poly(k/\eps)$ \cite{andoni2014towards} and $\tilde O(k/\eps^2)$ \cite{jambulapati2023sparsifying} (both obtained by sampling-based and therefore not contraction-based approaches), \Cref{main: lower} shows that contraction-based sparsifiers are sometimes provably suboptimal.

We then proceed to study the size upper bound for quality-$(1+\eps)$ contraction-based cut sparsifiers of quasi-bipartite graphs. Our last result, summarized in \Cref{quasi_apx}, shows that when $\eps$ is a constant, quasi-bipartite graphs have $\poly(k)$-sized contraction-based cut sparsifiers, complementing the lower bound in \Cref{main: lower}, in contrast with the $\Omega(2^k)$ lower bound for exact sparsifiers in \Cref{quasi_1}.

\begin{theorem}
\label{quasi_apx}
For every integer $k\ge 1$ and real number $\eps>0$, every quasi-bipartite graph with $k$ terminals admits a quality-$(1+\eps)$ contraction-based cut sparsifier of size $k^{O(1/\eps^2)}\cdot f(\eps)$, where $f$ is a function that only depends on $\eps$.
\end{theorem}

We summarize our results and provide some comparison with previous work in \Cref{table}.

\setlength{\tabcolsep}{0.5em} 
{\renewcommand{\arraystretch}{1.4}
\begin{table}[h]
\centering
	\begin{tabular}{|c|c|c|c|c|}
		\hline
		Graph Type                       & Quality      & Size                          & Contraction-based? & Citation                                                                                      \\  \hline
		\multirow{3}{*}{General }          & $1$          & $2^{2^k}$                & Yes            & \cite{hagerup1998characterizing,khan2014mimicking} \\ \cline{2-5}
		& $1$          & $2^{\Omega(k)}$               & No         &\cite{karpov2017exponential,khan2014mimicking}\\ 
		\cline{2-5}
		& $1$          & $2^{2^{\Omega(k)}}$               & Yes         &\cite{karpov2017exponential}\\\hline
		\multirow{3}{*}{Planar}          & $1$          & $2^{O(k)}$               & Yes            & \cite{krauthgamer2013mimicking,krauthgamer2017refined} \\ \cline{2-5}
		& $1$          & $2^{\Omega(k)}$               & No         &\cite{karpov2017exponential}\\
		\cline{2-5} 
		& $1+\eps$ & $O(k\cdot\poly(\log k/\eps))$ & No                & \Cref{main: upper}                                                           \\ \hline
		\multirow{3}{*}{Quasi-bipartite} & $1$          & $2^{O(k^2\log k)}, \Omega(2^{k})$                  & Yes               & \Cref{quasi_1}                                                              \\ \cline{2-5} 
		& $1+\eps$ & $\tilde O (k/\eps^2)$           & No                & \cite{jambulapati2023sparsifying}                                                    \\ \cline{2-5} 
		& $1+\eps$ & $k^{\tilde\Omega(1/\eps)}, k^{O(1/\eps^2)} f(\eps)$     & Yes               & \Cref{main: lower,quasi_apx}                                                           \\ \hline
	\end{tabular}
\caption{Results on exact or quality-$(1+\eps)$ cut sparsifiers for planar and quasi-bipartite graphs}
\label{table}
\end{table}
}

\subsection{Technical Overview}
\label{sec: tech_overview}

\paragraph{Planar cut sparsifiers.}
Our construction of planar cut sparsifiers follows a similar framework as the planar distance emulator construction in \cite{chang2022near}, based on the intuition that, if we take the dual graph, then min-cuts become min-length partitioning cycles and therefore can be approximated by distance-based approaches. The algorithm consists of three steps: 
\begin{enumerate}
\item reduce general planar graphs to $O(1)$-face instances (where terminals lie on $O(1)$ faces in the planar embedding of the graph);
\item reduce $O(1)$-face instances to $1$-face instances (where all terminals lie on the boundary of a single face; such graphs are also known as \emph{Okamura-Seymour instances}); and
\item construct quality-$(1+\eps)$ planar cut sparsifiers for $1$-face instances.
\end{enumerate}
Step $1$ and Step $3$ are implemented in a similar way as \cite{chang2022near}. Step $2$ requires several new ideas. The goal is to preserve the min-cuts for all terminal partitions. As terminals may lie on $O(1)$ faces, for different terminal partitions, the dual partitioning cycle corresponding to the min-cuts may go around these $O(1)$ faces in different ways, and therefore it is not enough to just preserve pairwise shortest paths in the dual graph. Indeed, as observed by \cite{krauthgamer2017refined}, we need to preserve, for each pair of terminals, and for each pattern (how a path going around the $O(1)$ faces, e.g., on the left side of face $1$, on the right side of face $2$, etc), the shortest path following the pattern and connecting the terminal pair. We call such a sparsifier a \emph{pattern distance emulator}, which a generalization of the standard distance emulators.
As there are $O(1)$ faces, there are $2^{O(1)}$ different patterns (left/right for each face). We manage to construct a near-linear size pattern emulator, incurring an unharmful factor $2^{O(1)}$ loss in its size. 

We believe that the notion of pattern emulators is of independent interest, and should play a role in solving other algorithmic problems on planar graphs.

\paragraph{Upper bounds for quasi-bipartite graphs.}
For exact (quality-$1$) contraction-based cut sparsifiers, a central notion in previous constructions and lower bound proofs is the \emph{profiles} of vertices. Specifically, the profile of a vertex $v$ specifies, for each terminal partition $(S,T\setminus S)$, on which side of the cut $\mc_G(S,T\setminus S)$ lies the vertex $v$.
Vertices of the same profile can be safely contracted together, and the number of different profiles is therefore an upper bound on the size of exact cut sparsifiers. For general graphs, a doubly-exponential upper bound of $2^{2^k}$ \cite{hagerup1998characterizing} on the number of distinct profiles was established, and this was later improved to $2^{\binom{k}{k/2}}$ \cite{khan2014mimicking}. For quasi-bipartite graphs, we will utilize its structure to proved a better (single-exponential) upper bound.

By the definition of quasi-bipartite graphs, there are no edge between non-terminal vertices, so non-terminal vertices form separate stars with terminals, and therefore contribute to terminal min-cuts independently. For example, in a star centered at $v$ with $6$ edges $e_1,\ldots,e_6$ connecting to terminal $t_1,\ldots,t_6$, respectively. The profile of vertex $v$ only depends on the weights of the edges $e_1,\ldots,e_6$: for the terminal partition $(S=\set{t_1,t_2,t_3},T\setminus S=\set{t_4,t_5,t_6})$, $$\text{$v$ lies on the $S$ side of $\mc_G(S,T\setminus S)$ iff } w(e_1)+w(e_2)+w(e_3)>w(e_4)+w(e_5)+w(e_6).$$

This gives us some power to reveal properties of profiles in quasi-bipartite graphs.
For example, if we consider subsets $S\subseteq T$ being $\set{1,2},\set{3,4},\set{1,3},\set{2,4}$, then $v$ cannot lie on the $S$ side for all these terminal min-cuts. Since otherwise, letting $w=\sum_{1\le  i\le 6}w(e_i)$, we get $w(e_1)+w(e_2)>w/2$, $w(e_3)+w(e_4)>w/2$, $w(e_1)+w(e_3)\le w/2$, and $w(e_2)+w(e_4)\le w/2$, a contradiction. This means that there are some ``configurations'' in the family of terminal subsets that are simply not realizable by any vertex profiles in quasi-bipartite graphs. We implement this idea through the notion of VC-dimension, realizing ``configurations'' by ``set systems'' and ``not realizable'' by ``non-shatterable'', and obtain a bound of $k^{O(k^2)}=2^{O(k^2\log k)}$ on the number of  profiles in quasi-bipartite graphs.

For quality-$(1+\eps)$ cut sparsifiers, since we allow a small multiplicative error in accuracy, we are not restricted to contracting vertices with the same profile, but are allowed to moderately manipulate vertex profiles so as to reduce their number. We implement this idea by ``projections onto $O(1/\eps^2)$-size stars''. On the one hand, consider for example a full star centered at $v$ (that is, a star containing edges $(v,t)$ for all $t\in T$) with uniform edge weights, and a subset $S\subseteq T$ with $|S|=|T|/2-1$, so $v$ lies on the $T\setminus S$ side of $\mc_G(S,T\setminus S)$. If we uniformly at random sample $c=O(1/\eps^2)$ edges from it and obtain a ``mimicking substar'' $H_v$, then by central limit theorem, with high probability $H_v$ contains at most $c/2+O(1/\eps)$ edges to $S$ and at least $c/2-O(1/\eps)$ edges to $T\setminus S$. Therefore, even if $H_v$ fails in mimicking the behavior of the full star on the cut $(S,T\setminus S)$, in that it mistakenly selects more edges to $S$ than to $T\setminus S$ and therefore place $v$ on the $S$ side of $\mc_G(S,T\setminus S)$ rather than the $T\setminus S$ side, it only causes an error of $\frac{c/2+O(1/\eps)}{c/2-O(1/\eps)}=1+O(\eps)$ in the min-cut size, which is allowed.

On the other hand, we have shown in the exact case that the number of profiles for stars on a fixed $O(1/\eps^2)$-size terminal set is $O(1/\eps^2)^{O(1/\eps^4)}$ (replacing $k$ with $O(1/\eps^2)$ in the $k^{O(k^2)}$ bound). Since the number of $O(1/\eps^2)$-size terminal subsets is $k^{O(1/\eps^2)}$, we can bound the total number of profiles produced by $O(1/\eps^2)$-size starts by 
 $k^{O(1/\eps^2)}\cdot O(1/\eps^2)^{O(1/\eps^4)}$, obtaining a size bound of $k^{O(1/\eps^2)} f(\eps)$ for quality-$(1+\eps)$ contraction-based cut sparsifiers for quasi-bipartite graphs.

\paragraph{Lower bounds for quasi-bipartite graphs.}
As shown in \cite{chen20241+,chen2024lower}, contraction-based sparsifiers are closely related to the Steiner node version of the classic $0$-Extension problem \cite{karzanov1998minimum}. Specifically, \cite{chen20241+} showed that the best quality achievable by contraction-based flow sparsifiers is bounded by the integrality gap of the semi-metric LP relaxation. For cut sparsifiers, we observe that the best achievable quality is controlled by an even more restricted case of the $0$-Extension with Steiner Nodes problem, where the underlying graph is a boolean hypercube. We focus on this special case, construct a hard hypercube-instance which is also a quasi-bipartite graph, and prove a size lower bound for its $(1+\eps)$-approximation, leading to a same size lower bound for quality-$(1+\eps)$ contraction-based cut sparsifiers for quasi-bipartite graphs.

\paragraph{Concurrent Work.} Independent of work, Das, Kumar, and Vaz, showed in their work \cite{das2024nearly} that quasi-bipartite graphs admit exact non-contraction-based cut sparsifiers of size $2^{k^2}$ and exact contraction-based cut sparsifiers of size $2^{k^3}$. Our \Cref{quasi_1} gives a $2^{O(k^2\log k)}$ size bound for exact contraction-based cut sparsifiers, which is slightly stronger than their $2^{k^3}$ bound, and slightly weaker than their $2^{k^2}$ bound. They also have some results on flow sparsifiers. For example, quasi-bipartite graphs admit exact flow sparsifiers of size $3^{k^3}$, and treewidth-$w$ graphs admit quality-$O(\frac{\log w}{\log\log w})$ flow sparsifiers of size $O(kw)$.

\subsection{Related Work}
\label{sec: related}

\paragraph{Edge sparsifiers.} 
After Karger's result \cite{karger1999random} on cut-preserving edge sparsifiers, there are other work using sampling-based approaches. Benzcur and Karger \cite{benczur1996approximate} sampled edges based on inverse edge-strengths, and obtained a sparsifier of $O(n\log n/\eps^2)$ edges. Fung and Harvey \cite{fung2010graph} sampled edges according to their inverse edge-connectivity, and also obtained the bound of $O(n\log^2 n/\eps^2)$. Spielman and Srivastava \cite{spielman2011graph} sampled edges based on their inverse effective-resistance to obtain a spectral sparsifier (a generalization of cut sparsifiers) with size $O(n\log n/\eps^2)$. This bound was later improved by Batson, Spielman, and Srivastava \cite{batson2012twice} to $O(n/\eps^2)$.

\paragraph{Other work on cut sparsifiers.}
There are some work on (i) constructing better cut sparsifiers for  special families of graphs, for example trees \cite{goranci2017vertex} and planar graphs with all terminals lying on the same face \cite{goranci2017improved}, and graphs with bounded treewidth \cite{andoni2014towards}; (ii) preserving terminal min-cut values up to some threshold value \cite{chalermsook2021vertex,liu2020vertex}; and
(iii) dynamic cut/flow sparsifiers and their utilization in dynamic graph algorithms \cite{durfee2019fully,chen2020fast,goranci2021expander}.

\paragraph{Flow sparsifiers.}
Flow sparsifiers are highly correlated with cut sparsifiers.
Given a graph $G$ and a set $T$ of $k$ terminals, a graph $G'$ is a \emph{flow sparsifier} of $G$ with respect to $T$ with \emph{quality $q$}, iff every $G$-feasible multicommodity flow on $T$ can be routed in $G'$, and every $G'$-feasible multicommodity flow on $T$ can be routed in $G$ if the capacities of edges in $G$ are increased by factor $q$. Flow sparsifiers are stronger than cut sparsifiers, in the sense that quality-$q$ flow sparsifiers are automatically quality-$q$ cut sparsifiers, but the other direction is not true in general.

When no Steiner nodes are allowed,
Leighton and Moitra \cite{leighton2010extensions} showed the existence of quality-$O(\log k/\log\log k)$ flow sparsifiers. 
On the lower bound side, the first quality lower bound in \cite{leighton2010extensions} is $\Omega(\log\log k)$, and this was later improved to $\Omega(\sqrt{\log k/\log\log k})$ by Makarychev and Makarychev \cite{makarychev2010metric}.

In the setting where Steiner nodes are allowed,
Chuzhoy \cite{chuzhoy2012vertex} showed $O(1)$-quality contraction-based flow sparsifiers with size $C^{O(\log\log C)}$ exist, where $C$ is the total capacity of all edges incident to terminals (assuming that all edges have capacity at least $1$).
On the lower bound side, Krauthgamer and Mosenzon \cite{krauthgamer2023exact} showed that there exist $6$-terminal graphs whose quality-$1$ flow sparsifiers must have an arbitrarily large size (i.e., the size bound cannot depend only on $k$ and $\eps$). Chen and Tan \cite{chen20241+} showed that there exists $6$-terminal graphs whose quality-$(1+10^{-18})$ contraction-based flow sparsifiers must have an arbitrarily large size.

\subsection{Organization}
The rest of the paper is organized as follows. We start with some preliminaries and formal definitions in \Cref{sec: prelim}.
We first provide the construction of planar cut sparsifier in \Cref{sec: planar}, proving \Cref{main: upper}. We then show the construction of contraction-based cut sparsifiers for quasi-bipartite graphs in \Cref{sec: quasi_exact} and \Cref{sec: quasi_apx}, proving \Cref{quasi_1} and \Cref{quasi_apx}. Finally, we show the lower bound for quasi-bipartite graphs, giving the proof of \Cref{main: lower} in \Cref{sec: lower}.

\section{Preliminaries}
\label{sec: prelim}

By default, all logarithms are to the base of $2$. 
All graphs are undirected and simple.

Let $G=(V,E,w)$ be an edge-weighted graph, where each edge $e\in E$ has weight $w_e$. 
In this paper, weights can be either representing capacities (in the context of cut values) or lengths (in the context of shortest-path distances).
For a vertex $v\in V$, we denote by $\deg_G(v)$ the degree of $v$ in $G$.
For a pair $v,v'$ of vertices in $G$, we use $\dist_{G}(v,v')$ (or $\dist_{w}(v,v')$) to denote the shortest-path distance between $v$ and $v'$ in $G$, viewing edge weights as lengths.
We may omit the subscript $G$ in the above notations when the graph is clear from the context.

\paragraph{Cut sparsifiers, contraction-based sparsifiers.}
Let $G$ be a graph and let $T$ be a subset of its vertices called \emph{terminals}. 
Let $H$ be a graph with $T\subseteq V(H)$. We say that $H$ is a \emph{cut sparsifier} of $G$ with respect to $T$ with \emph{quality} $q\ge 1$, iff for any partition of $T$ into two non-empty subsets $T_1,T_2$, 
\[\mc_H(T_1,T_2)\le \mc_G(T_1,T_2)\le q\cdot\mc_H(T_1,T_2).\]

We will use the following observation, which is immediate from the definition of cut sparsifiers.
\begin{observation}
\label{obs: chain}
If $(G_1,T)$ is a quality-$q_1$ cut sparsifier of $(G,T)$, and $(G_2,T)$ is a quality-$q_2$ cut sparsifier of $(G_1,T)$, then $(G_2,T)$ is a quality-$q_1q_2$ cut sparsifier of $(G,T)$.
\end{observation}

We say that a graph $H$ is a \emph{contraction-based sparsifier of $G$ with respect to $T$}, iff there exists a partition $\lset$ of vertices in $G$ into subsets where different terminals in $T$ lie in different subsets in $\lset$, and $H$ is obtained from $G$ by contracting, for each $L\in\lset$, vertices in $L$ into a supernode $u_L$, keeping parallel edges and discarding self-loops. 

\section{Proof of \Cref{main: upper}}
\label{sec: planar}

\newcommand{\gluepath}{{\mathsf{Glue}}}
\newcommand{\cutpath}{{\mathsf{Cut}}}

In this section we provide the proof of \Cref{main: upper}. We will use a similar algorithmic framework as \cite{chang2022near}, by first considering the case where all terminals lie on the boundary of one face and then generalizing the arguments to the $O(1)$-face case and eventually to the general case.

Let $G,G'$ be planar graphs with $T\subseteq V(G),V(G')$.
Let $\fset$ be the set of faces in the planar drawing of $G$ that contain all terminals, and we define $\fset'$ similarly for $G'$.
We say instances $(G,T),(G',T)$ are \emph{aligned}, iff there is a one-to-one correspondence between sets $\fset$ and $\fset'$, such that for each face $F\in \fset$, the set of terminals in $G$ lying on $F$ is exactly the set of terminals in $G'$ lying on $F'$, its corresponding face in $\fset'$, and the circular ordering in which these terminals appear on $F$ and $F'$ are also identical.
We say $(G,T)$ is a $|\fset|$-face instance. Throughout this section, all the cut sparsifiers that we will ever construct for the input instance or any subinstance generated from it are aligned.

\subsection{Preparation: removing separator terminals, and dual graphs}

We start with the following divide and conquer lemma, whose proof is deferred to \Cref{apd: Proof of lem: divide}.

\begin{lemma}
\label{lem: divide}
Let $(G,T)$ be a planar instance and let $\gset$ be a collection of subgraphs of $G$ such that the edge sets $\set{E(G') \mid G'\in \gset}$ partition $E(G)$. For each graph $G'\in \gset$, we denote by $T_{G'}$ the set that contains all vertices in $T\cap V(G')$ and all vertices of $G'$ that appear in some other graph in $\gset$.
If for each $(G',T_{G'})$ we are given a quality-$q$ aligned cut sparsifier $(H_{G'}, T_{G'})$, then we can efficiently obtain a quality-$q$ planar cut sparsifier $H'$ of $G$ with $|V(H')|\le \sum_{G'\in \gset}|V(H_{G'})|$.
\end{lemma}

We say a vertex $v\in V(G)$ is a \emph{separator} iff $G\setminus v$ is disconnected.
Denote its connected components by $G'_1,\ldots,G'_r$.
For each $1\le i\le r$, let $G_i$ be the subgraph of $G$ induced by $V[G'_i]\cup \set{v}$. We say graphs $G_1,\ldots,G_r$ are obtained by \emph{splitting $G$ at $v$}.
We use the following lemma from \cite{chang2022near}.

\begin{claim}[Claim 4.3 in \cite{chang2022near}]
\label{clm: separator}
Let $(G,T)$ be an instance.
Let $U$ be the set of separators in $T$. Let $\gset$ be the set of graphs obtained by splitting $G$ at all vertices in $U$ (that is, we first split $G$ at a vertex in $U$, and then repeatedly split each of the obtained graph at one of its separator terminal, until no resulting graph contains any separator terminal). Then
$\sum_{G'\in \gset}|T\cap V(G')|\le O(|T|)$.
\end{claim}

With \Cref{lem: divide}, we will assume that no terminal in $T$ is a separator in $G$, and so if we go around the boundary of a face in $\fset$, we will visit each terminal that lies on this face exactly once. This is since we can split $G$ at all its separator terminals, compute quality-$(1+\eps)$ cut sparsifiers for them separately, and then combine them in the way of \Cref{lem: divide}, paying only an $O(1)$ factor in the total size, as guaranteed by \Cref{clm: separator}.

\paragraph{Dual graphs.}
Cuts in planar graphs are cycles in their dual graphs. We will exploit this property in a similar way as \cite{krauthgamer2017refined} and transform the cut-preserving tasks to distance-preserving tasks. For technical reasons, in order to handle graphs with terminals, our definition of dual graphs are slightly different from the standard planar dual graphs.

Let $(G,T)$ be the input instance and let $\fset$ be the faces that contains all terminals.
Assume each face is a simple cycle. For each $F\in \fset$, add an auxiliary vertex $x_F$, and connect it to every terminal via a new edge, such that these edges are drawn in an internally disjoint way within face $F$, so that they separates face $F$ into subfaces, which we call special subfaces. Then we take the standard planar dual of the modified graph.
Here each special face correspond to a node, which we call \emph{dual terminals}. We then remove all edges in the dual graph connecting a pair of dual terminals. We denote the resulting graph by $G^*$, and denote by $T^*$ the set of dual terminals. See \Cref{fig: dual} for an illustration.
Note that when graph $G$ does not have the property that each face is a simple cycle, as long as we are guaranteed that no terminal is a separator, we can still define dual graphs similarly.
It is also easy to verify that the dual of the dual graph $G^*$ is the original graph $G$ itself. Sometimes we say that we \emph{reverse} $G^*$ to obtain $G$.
Finally, as both $G$ and $G^*$ are planar graphs, and each non-terminal vertex in $G^*$ corresponds to a face in $G$, the number of vertices in $G$ and $G^*$ are within $O(1)$ factor from each other.

\begin{observation}
	$(G,T)$ is a $f$-face instance iff $(G^*,T^*)$ is a $f$-face instance.
\end{observation}
We denote by $\fset^*$ the faces in $G^*$ that contain all dual terminals. It is easy to see that there exists a natural one-to-one correspondence between $\fset$ and $\fset^*$.

\begin{figure}[h]
	\centering
	\subfigure
	{\scalebox{0.1}{\includegraphics{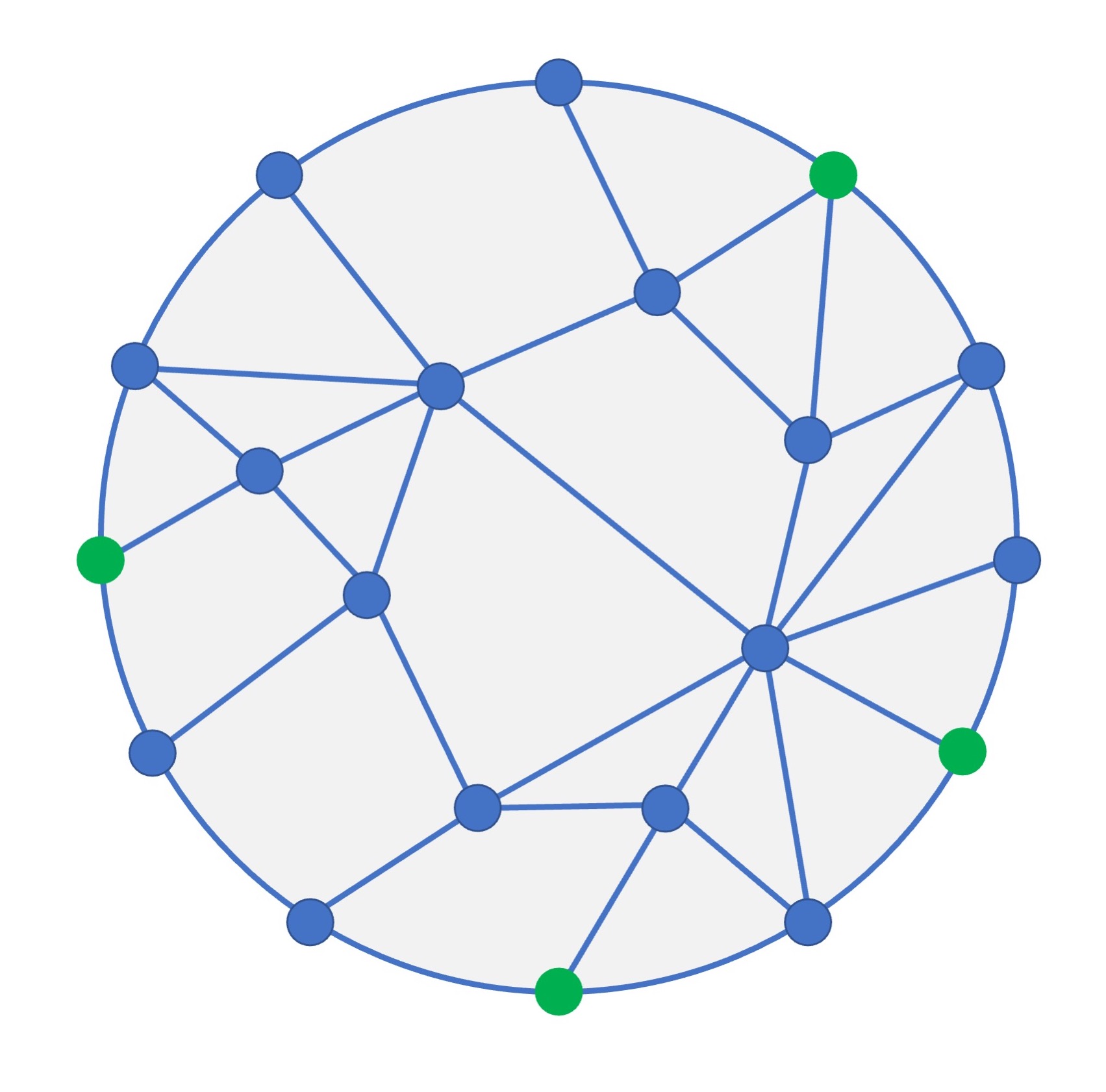}}}
	\hspace{0.7cm}
	\subfigure
	{
		\scalebox{0.1}{\includegraphics{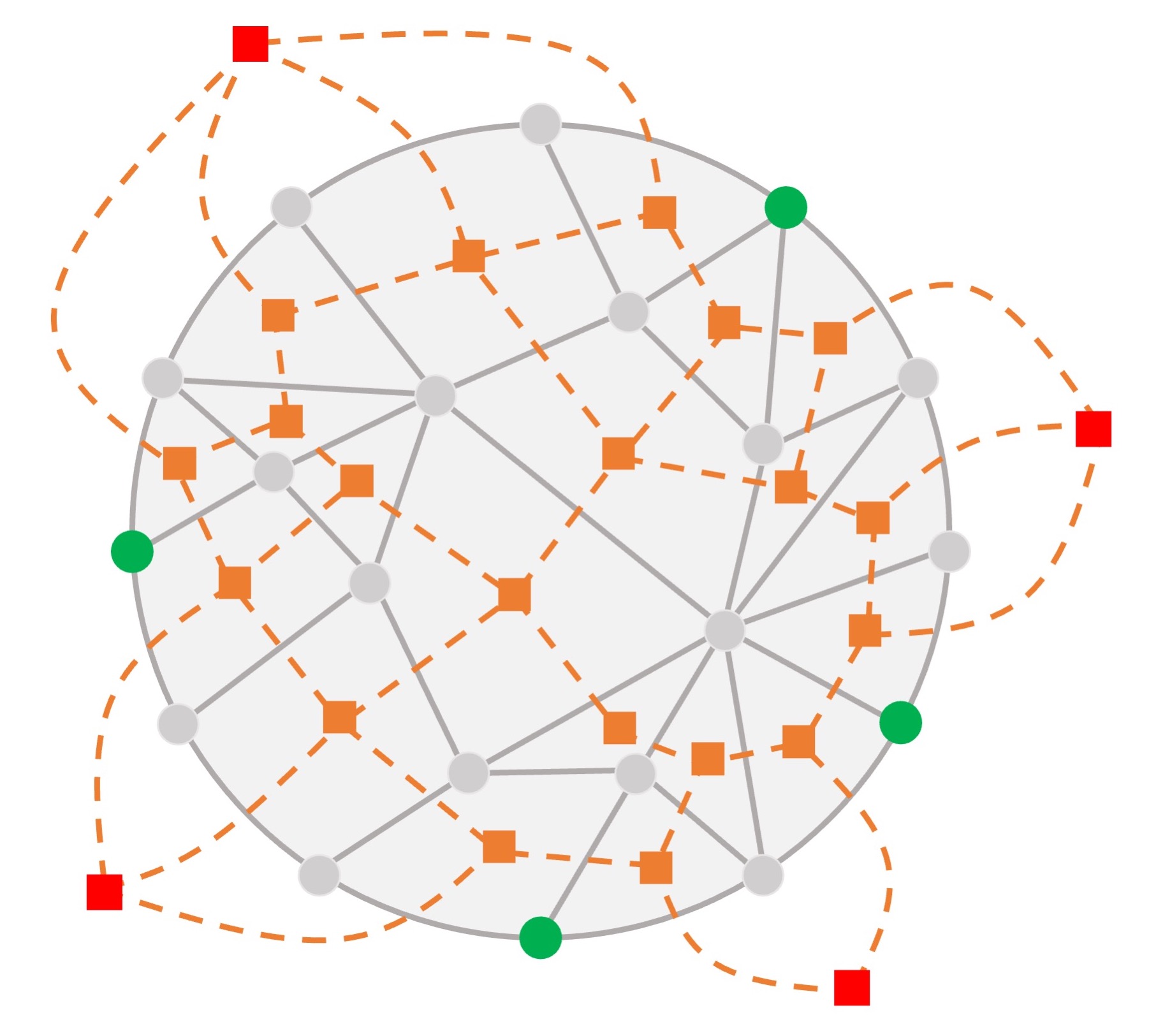}}}
	\caption{An illustration of a one-face instance $(G,T)$ (left) and its dual $(G^*,T^*)$ defined in our way (right). The dual terminals in $T^*$ are marked in red. Taking the dual of (or reverse) $G^*$ we get $G$ back.\label{fig: dual}}
\end{figure}

\subsection{Warm-up: the one-face case}
\label{sec: 1-face}

In this subsection we prove \Cref{main: upper} for the special case where all terminals lie on the boundary of the outer face in the planar embedding of $G$.

\begin{lemma}
	\label{lem: mincut structure}
Let $(G,T)$ be a one-face instance. Then for any partition $(T_1,T_2)$ of the terminals into non-empty subsets, for any min-cut $\hat E$ in $G$ separating $T_1$ from $T_2$, the corresponding edge set $\hat E^*$ in the dual graph $G^*$ is the union of several edge-disjoint shortest paths connecting dual terminals.
\end{lemma}

\begin{proof}
	For every non-terminal vertex $v$ in $G$, we say that it is a \emph{$1$-node} iff in graph $G\setminus \hat E$, $v$ lies in the same connected component with some terminal in $T_1$, and we define \emph{$2$-nodes} similarly.
	
	\begin{observation}
		\label{obs: 1 vs 2}
		Every non-terminal $v$ in $G$ is either a $1$-node or a $2$-node, but not both, and every edge in $\hat E$ connects a $1$-node to a $2$-node.
	\end{observation}
	\begin{proof}
		As $\hat E$ is a cut separating terminals in $T_1$ from terminals in $T_2$, in $G\setminus \hat E$ every non-terminal cannot be both a $1$-node and a $2$-node, as otherwise some $T_1$ terminal lies in the same connected component with some $T_2$ terminal, a contradiction.
		
		Assume now that there is a vertex that is neither a $1$-node nor a $2$-node. Consider the connected component $C$ of $G\setminus\hat E$ that contains it, so $C$ contains no terminals. Consider any edge $e$ connecting an vertex in $C$ to a vertex $v'$ outside $C$, clearly $e$ belongs to $\hat E$.
		We claim that $\hat E\setminus e$ is also a cut separating $T_1$ from $T_2$, contradicting the minimality of $\hat E$.
		In fact, if $v'$ is a $1$-node ($2$-node, resp.), then in $G\setminus (\hat E\setminus e)$ all vertices of $C$ are also $1$-nodes ($2$-nodes, resp.); and if $v'$ is neither a $1$-node nor a $2$-node, then in $G\setminus (\hat E\setminus e)$ all vertices of $C$ are neither $1$-nodes nor $2$-nodes. In either case, $T_1$ are still separated from $T_2$ in $G\setminus (\hat E\setminus e)$.
		
		It is now easy to see that every edge in $\hat E$ connects a $1$-node to a $2$-node, since otherwise we can remove such an edge from $\hat E$, still separating $T_1$ from $T_2$, contradicting the minimality of $\hat E$.
	\end{proof}

	Denote $T=\set{t_1,\ldots,t_k}$, where the terminals are indexed according to the order that they appear on the boundary of the outer face. An \emph{interval} is a subsequence $t_i,t_{i+1}\ldots,t_j$, such that all terminals $t_i,t_{i+1}\ldots,t_j$ lies in the same connected component in $G\setminus \hat E$ that does not contain $t_{i-1}, t_{j+1}$.
	We say that an interval is a $1$-interval ($2$-interval, resp.) if all its terminals lie in $T_1$ ($T_2$, resp.).
	
	\begin{observation}
		\label{obs: interval}
		There is an interval $T'$, such that in $G\setminus \hat E$, $T'$ is separated from $T\setminus T'$.
	\end{observation}
	\begin{proof}
		We find such an interval as follows. Consider any interval $T'$. If it is separated from $T\setminus T'$ in $G\setminus \hat E$, then we are done. If not, then it is connected to some other interval $T''$ in $G\setminus \hat E$. They separate the boundary of the outer face into two segments both containing some terminals, but the terminals in one segment are not connected to the terminals in the other segment. We then recurse on one of these segments, start with any interval. Eventually we will find a desired interval.
	\end{proof}
	
	Consider an interval $T'$ and let $C$ be the connected component in $G\setminus \hat E$ that contains $T'$. Assume without loss of generality that $T'\subseteq T_1$. Let $\delta(C)$ be the edges in $G$ with exactly one endpoint in $C$, so clearly $\delta(C)\subseteq \hat E$. 
	
	\begin{claim}
		\label{clm: interval cut}
		$\delta(C)$ is the min-cut in $G$ separating $T'$ from $T\setminus T'$.
	\end{claim}
	\begin{proof}
		Let $E'$ be the mincut in $G$ separating $T'$ from $T\setminus T'$. We show that $(\hat E\setminus \delta(C))\cup E'$ is a cut in $G$ separating $T_1$ from $T_2$. From the minimality of $\hat E$, this means that $w((\hat E\setminus \delta(C))\cup E')\ge w(\hat E)$, implying that $w(\delta(C))\le w(E')$, or equivalently, $\delta(C)$ is the min-cut in $G$ separating $T'$ from $T\setminus T'$.
		In fact, from \Cref{obs: 1 vs 2}, all outside-of-$C$ endpoints of edges in $\delta(C)$ are $2$-nodes.
		Therefore, it is easy to verify that $\hat E\setminus \delta(C)$ is a cut in $G$ separating $T_1\setminus T'$ to $T_2\cup T'$. As $E'$ is a cut separating $T'$ from $T_2$, we conclude that $\hat E\setminus \delta(C))\cup E'$ separates $T_1$ from $T_2$.
	\end{proof}

	\Cref{clm: interval cut} implies that, in the dual graph $G^*$, the dual edges of edges in $\delta(C)$ form a shortest path connecting a pair of terminals in $T^*$.
	In other words, we have managed to extract part of the cut $\hat E$ that is a shortest path in $G^*$. We now recursively extract other parts using the following claim from the rest of $\hat E$.
	
	\begin{claim}
		\label{clm: recursive}
		$\hat E\setminus \delta(C)$ is a min-cut in $G$ separating $T_1\setminus T'$ from $T_2$.
	\end{claim}
	\begin{proof}
		On one hand, we have shown in the proof of \Cref{clm: interval cut} that $\hat E\setminus \delta(C)$ is a cut in $G$ separating $T_1\setminus T'$ to $T_2\cup T'$, so $\hat E\setminus \delta(C)$ does separate $T_1\setminus T'$ from $T_2$.
		On the other hand, if we let $E'$ be a min-cut in $G$ separating $T_1\setminus T'$ from $T_2$, then via similar analysis as in proof of \Cref{clm: interval cut}, we can show that $\delta(C)\cup E'$ is a cut in $G$ separating $T_1$ from $T_2$, implying that $w(\hat E)\le w(\delta(C)\cup E')$ and so $w(\hat E\setminus \delta(C))\le w(E')$.
	\end{proof}
	According to \Cref{clm: recursive}, we are able to repeatedly extract subsets of edges in $\hat E$ that correspond to shortest paths in $G^*$, using \Cref{obs: interval} and \Cref{clm: interval cut}, until $\hat E$ becomes empty. This completes the proof of \Cref{lem: mincut structure}. 
\end{proof}

We are now ready to prove \Cref{main: upper} for one-face instances using \Cref{main: upper}.
We first construct the dual graph $G^*$, and use the result in \cite{chang2022near} to construct a quality-$(1+\eps)$ planar distance emulator $\hat G^*$ with size $O(k\cdot \poly(\log k/\eps))$, and then reverse it to get $\hat G$. 
From our definition of dual graphs and their reverse, $\hat G$ contains all terminals and $|V(\hat G)|=O(k\cdot \poly(\log k/\eps))$. It remains to show that $(\hat G,T)$ is a quality-$(1+\eps)$ cut sparsifier for $(G,T)$. 

Let $(T_1,T_2)$ be a partition of $T$.
On the one hand, the min-cut in $G$ separating $T_1$ from $T_2$ consists of several edge-disjoint dual terminal shortest paths in $G^*$, and as $\hat G^*$ is a $(1+\eps)$ aligned distance emulator for $G^*$, we can find, for each such dual terminal shortest path in $G^*$, a shortest path in $\hat G^*$ connecting its corresponding dual terminals, with at most $(1+\eps)$ times the length. By taking the union of the edges in all these paths in $\hat G^*$ (note however that these paths are not necessarily edge-disjoint in $\hat G^*$), we obtain an edge set in $\hat G$ separating $T_1$ from $T_2$ whose value is at most $(1+\eps)$ times the value of $\mc_G(T_1,T_2)$. 
On the other hand, we apply the same analysis, starting with a min-cut in $\hat G$ separating $T_1$ from $T_2$. As $G^*$ is also a $(1+\eps)$ aligned distance emulator for $\hat G^*$, we can derive similarly that $\mc_{\hat G}(T_1,T_2)\le (1+\eps)\cdot \mc_G(T_1,T_2)$.

\subsection{The $O(1)$-face case}

In this subsection we prove the following lemma.

\begin{lemma}
	\label{lem: O(1) face}
	There exists a universal constant $c>0$, such that for any real number $0<\eps<1$, every $f$-face instance $(H,U)$ has an aligned $f$-face instance $(H',U)$ as its quality-$(1+\eps)$ cut sparsifier, with $|V(H')|\le |U|\cdot (cf\log |U|/\eps)^{cf^2}$.
\end{lemma}

Terminal-separating min-cuts in $G$ dual-terminal shortest paths in our graph $G^*$. Therefore, from now on we focus on preserving dual-terminal shortest paths in $G^*$. Recall that $\fset^*$ is the collection of faces in $G^*$ that contains all dual terminals. The following notion is central in our algorithm.

\paragraph{Pattern-shortest paths, pattern distances, and pattern emulators.}
For a pair $t,t'$ of dual terminals in graph $G^*$, there are more than one way for them to be connected by a path, depending on how the path goes around other faces in $\fset^*$, called \emph{patterns},  A rigorous way of defining patterns was given in \cite{krauthgamer2017refined}, which we describe below.

Fix a point $\nu^*$ in the plane.
Let $\gamma$ be a closed curve (not necessarily simple). A vertex $v\notin \gamma$ is said to be \emph{inside} $\gamma$ iff any simple curve connecting $\nu^*$ to $v$ crosses $\gamma$ an odd number of times, otherwise it is said to be \emph{outside} $\gamma$.
For a pair $t,t'$ of dual terminals in $G^*$, we denote by $\Pi_{t,t'}$ the shortest path connecting them in $G^*$.
For every face $F\in \fset^*$, we place a point $z_F$ inside its interior. Consider now a path $P$ in $G^*$ connecting $t,t'$. Note that the union of $P$ and $\Pi_{t,t'}$ is a closed curve. Now the \emph{pattern} of $P$ is defined to be a $|\fset^*|$-dimensional vector $\pat(P)=(\phi_F)_{F\in \fset^*}$, where $\phi_F=1$ iff $z_F$ is inside the closed curve formed by the union of $P$ and $\Pi_{t,t'}$ and $\phi_F=-1$ iff $z_F$ is outside.

\begin{observation}
\label{obs: pattern}
Between every pair of terminals there are $2^{f}$ different patterns.
\end{observation}
\begin{proof}
As a pattern is a $|\fset^*|=|\fset|=f$-dimensional $+/-$ vector. There are $2^f$ possibilities.
\end{proof}

Let $P,P'$ be simple paths with same endpoints, the above definition of pattern implies that paths $P,P'$ \emph{have the same pattern} iff all points $\set{z_F}_{F\in \fset^*}$ lie outside the closed curve $P\cup P'$.
We will also consider paths $P,P'$ with one common endpoint $t$. Let $R$ be a path such that the other-endpoints of $P$ and $P'$ both lie on $R$. In this case we say that paths $P,P'$ have the \emph{same pattern with respect to $R$}, iff all points $\set{z_F}_{F\in \fset^*}$ lie ourside the closed curve formed by the union of $P$, $P'$, and the subpath of $R$ between the other-endpoint of $P$ and the other-endpoint of $P'$.

For a pattern $\pat$, the $t$-$t'$ \emph{$\pat$-shortest path} is defined to be the shortest path among all $t$-$t'$ paths with pattern $\pat$. Its length is defined to be the \emph{$\pat$-distance} between $t,t'$, which we denote by $\dist^{\Phi}(t_1,t_2)$.
Let $(G,T), (H,T)$ be aligned instances. We say that $(H,T)$ is a \emph{quality-$q$ pattern emulator} of $(G,T)$, iff for every pair $t_1,t_2$ of terminals and every pattern $\pat$, the $\pat$-distance between $t_1,t_2$ in $G$ is within factor $q$ from the $\pat$-distance between $t_1,t_2$ in $H$.
We use the following result from \cite{krauthgamer2017refined}.

\begin{theorem}[Adaptation of Theorem 4.6 in \cite{krauthgamer2017refined}]
\label{thm: pattern for cut}
Let $(G,T),(H,T)$ be aligned instances. If the dual instance $(H^*,T^*)$ is a quality-$q$ pattern emulator of the dual instance $(G^*,T^*)$, then $(H,T)$ is a quality-$q$ cut sparsifier of $(G,T)$.
\end{theorem}

\subsection{Constructing pattern emulators}

We will now prove the following lemma, which, combined with \Cref{thm: pattern for cut}, implies \Cref{lem: O(1) face}.

\begin{lemma}
\label{lem: O(1) face emulator}
There exists a universal constant $c>0$, such that for any $0<\eps<1$, every $f$-face instance $(G,T)$ admits a quality-$(1+\eps)$ pattern emulator $(H,T)$ with $|V(H)|\le |T|\cdot (cf\log |T|/\eps)^{cf^2}$.
\end{lemma}

We prove \Cref{lem: O(1) face emulator} by induction on $f$. The base case where $f=1$ has been proved in \Cref{sec: 1-face}.
Assume that we are given a $f$-face instance $(G,T)$.
We will convert it into an $(f-1)$-face instance and apply the inductive hypothesis.
Throughout this subsection, we use the parameters $\eps'=\eps/2f^2$ and $\eps''=(1-1/f^2)\cdot \eps$.

We first pick a pair $t,t'$ of terminals that lie on distinct faces, denoted by $\alpha,\alpha'$, respectively, and compute the shortest path $\Pi_{t,t'}$ in $G$ connecting $t$ to $t'$, such that $\Pi_{t,t'}$ does not contain any terminal as inner vertices. Denote $\Pi=\Pi_{t,t'}$. Then we compute a set of vertices on $\Pi$ that we call \emph{portals}. 

\paragraph{Portals on $\Pi$.}
We use the notion of $\eps$-covers defined in \cite{klein1998fully,thorup2004compact}.
Let $R$ be a shortest path and let $v$ be a vertex that does not lie in $R$. An \emph{$\eps$-cover} of $v$ in $R$ is a set $C(v,R)$ of vertices, such that for any vertex $x\in R$, there exists some $y\in C(v,R)$, such that 
$\dist(v,x)\le \dist(v,y)+\dist(y,x)\le (1+\eps)\cdot\dist(v,x).$
It has been shown \cite{klein1998fully,thorup2004compact} that for every $\eps>0$, shortest path $R$ and vertex $v\notin R$, there exists an $\eps$-cover of $v$ in $R$ of size $O(1/\eps)$.
As we want to construct pattern emulators, we will need to modify the definition of $\eps$-covers into a ``pattern version'' accordingly, as follows.

Let $G$ be an $f$-face instance, $R$ be a shortest path in $G$ that does not intersect any face internally, and $v$ be a vertex that does not lie in $R$. 
Let $Q$ be a path connecting $v$ to some vertex in $R$, and let $\Phi$ be its pattern. A \emph{$\Phi$-respecting $\eps$-cover} of $v$ in $R$ is a set $C(v,R,\Phi)$ of vertices, such that for any vertex $x\in R$, there exists some $y\in C(v,R,\Phi)$, such that the shortest $v$-$y$ path with the same pattern as $\Phi$, concatenated with the $y$-$x$ subpath of $R$, is within factor $(1+\eps)$ in length with the shortest $v$-$x$ path with the same pattern as $\Phi$. In other words (slightly abusing the notations),
$$\dist^\Phi(v,x)\le \dist^\Phi(v,y)+\dist(y,x)\le (1+\eps)\cdot\dist^\Phi(v,x).$$
We use the following lemma, whose proof is similar to the previous result of $\eps$-cover and is deferred to \Cref{apd: Proof of lem: pattern cover}.

\begin{lemma}
\label{lem: pattern cover}
For any planar instance $G$, shortest path $R$ in $G$, vertex $v\notin R$, and pattern $\Phi$, there exists a $\Phi$-respecting $\eps$-cover of $v$ in $R$ of size $O(1/\eps)$.
\end{lemma}

We now define a set $Y$ of vertices on $\Pi$ as follows. For each terminal $t$ and for each pattern $\Phi$, let $Y_{t,\Phi}$ be the pattern-respecting $\eps'$-cover of $v$ in $R$ of size $O(1/\eps')$ given by \Cref{lem: pattern cover}. We then let set $Y=\bigcup_{t,\Phi}Y_{t,\Phi}$.
From \Cref{obs: pattern} and \Cref{lem: pattern cover}, 
$|Y|\le O(2^f\cdot |T|/\eps')$.

We then slice the graph $G$ open along the path $\Pi$ connecting faces $\alpha,\alpha'$.
Specifically, we duplicate path $\Pi$ into two copies $\Pi_1,\Pi_2$ and place them closely at two sides (which we call $1$-side and $2$-side, respectively) of its original image, such that the thin strip between them connect faces $\alpha,\alpha'$ into a single face $\beta$.
All vertices $x$ on $P$ into their copies $x_1,x_2$ ($x_1$ lies on $\Pi_1$ and $x_2$ lies on $\Pi_2$), including the terminals $t,t'$ that now have copies $t_1,t_2$ and $t'_1,t'_2$. All weights on edges remain the same. 
It is easy to verify that the obtained graph, which we denote by $G'$, is a $(f-1)$-face instance.
Let $T'$ be the set that contains all copies of terminals and portals in $Y$.
We then construct a quality-$(1+\eps'')$ pattern emulator $H'$ of $G'$ with respect to $T'$, and then glue, for each portal $y\in Y$, its copies $y_1,y_2$ back to their original vertex $y$. The resulting graph is denoted by $H$ and is returned as the pattern emulator for $G$.
For a complete description of the cut and glue operations, please refer to Appendix B of \cite{chang2022near}.
See \Cref{fig: splitting h-hole after} for an illustration.

\begin{figure}[h!]
	\centering
	\subfigure[Graph $G$: faces $\alpha, \alpha'$ (shaded gray), path $\Pi$ (blue), and portals (purple). ]{\scalebox{0.5}{\includegraphics[scale=0.15]{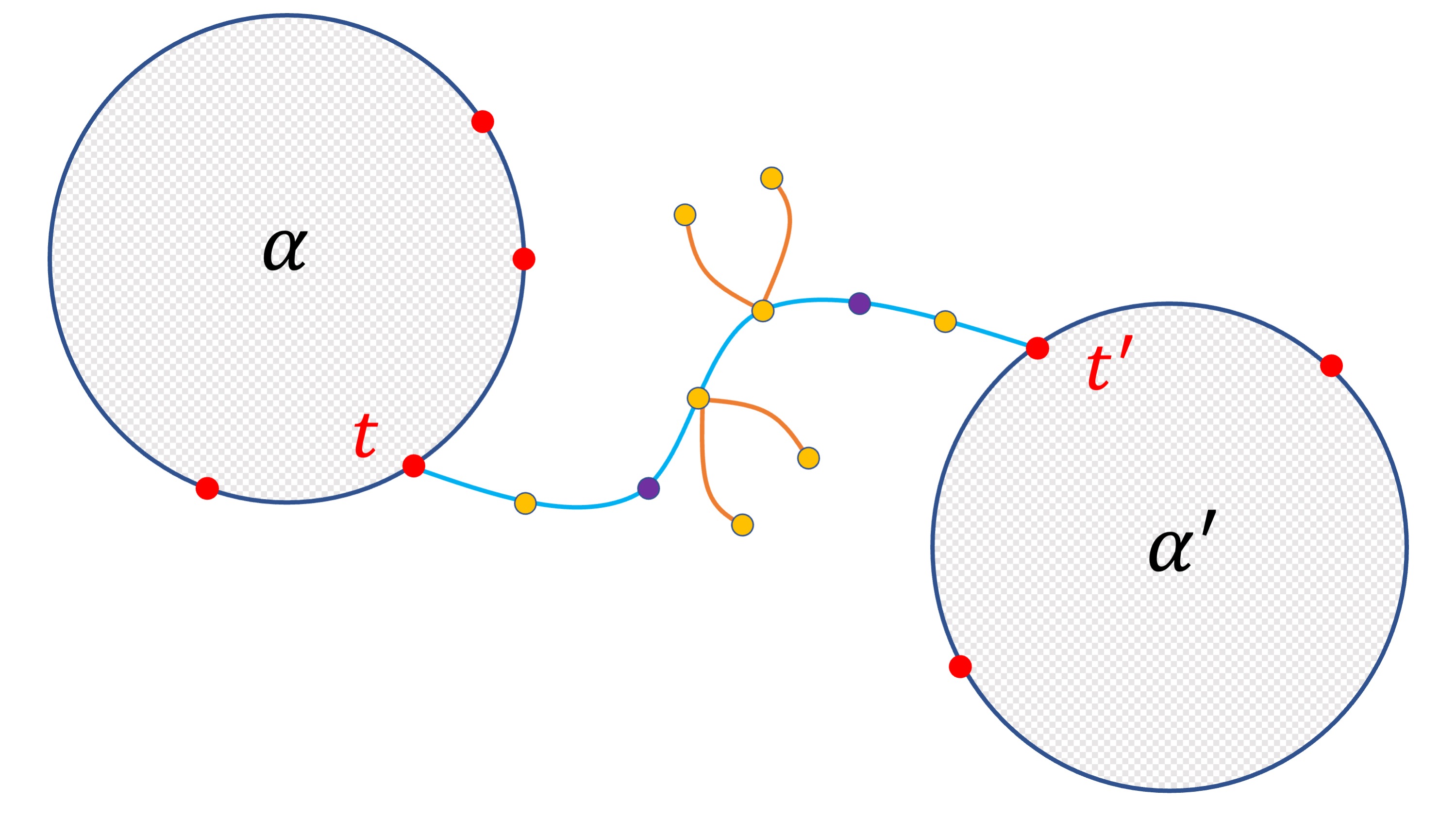}\label{fig: splitting h-hole before}}}
	\hspace{0.2cm}
	\subfigure[Graph $G'$: the new hole $\beta$ (shaded gray), and paths $\Pi_1,\Pi_2$ (blue).]{\scalebox{0.5}{\includegraphics[scale=0.15]{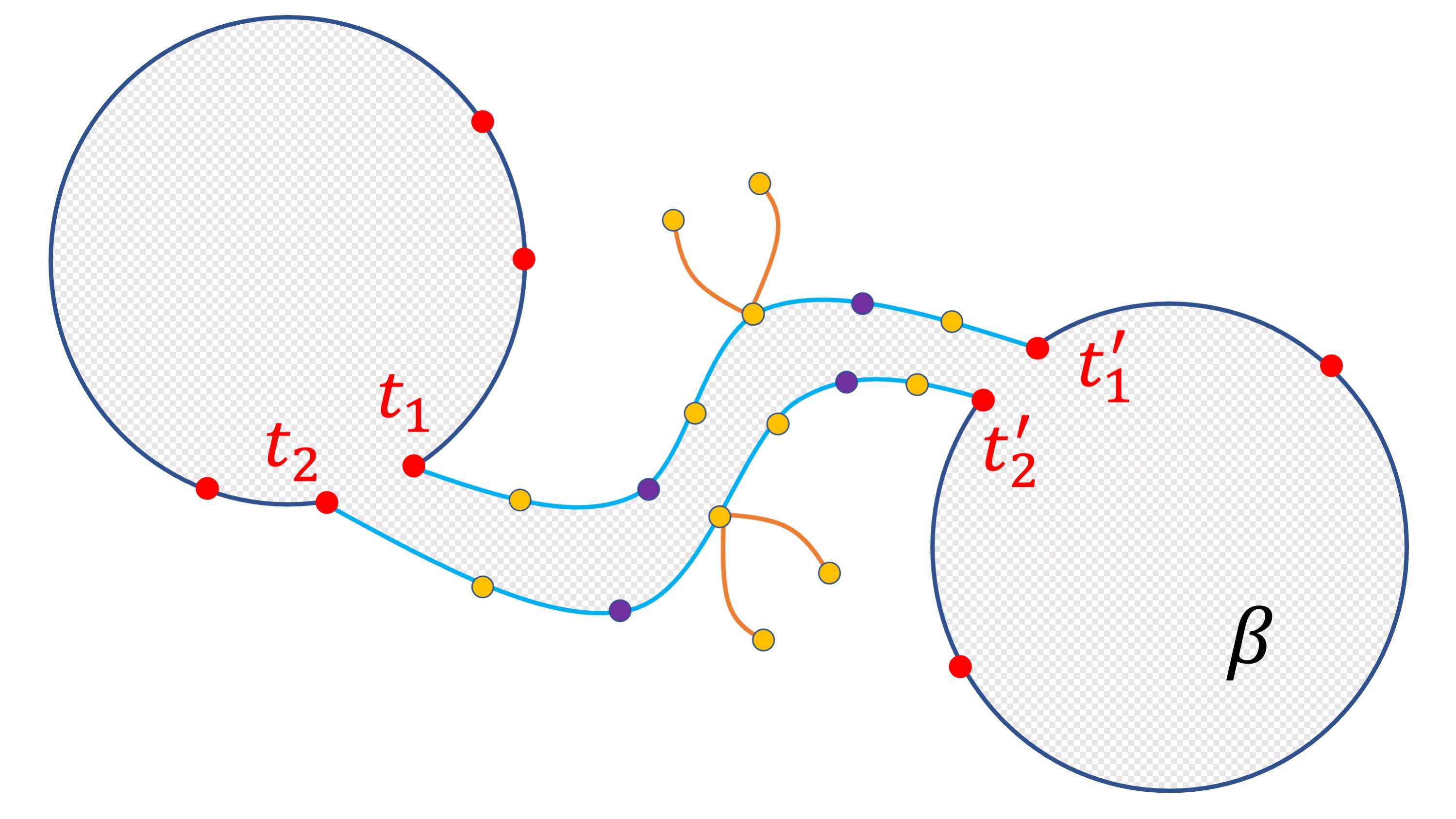}\label{fig: splitting h-hole after}}}
	\subfigure[Graph $H$: holes $\alpha,\alpha'$ and portals are restored.]{\includegraphics[scale=0.08]{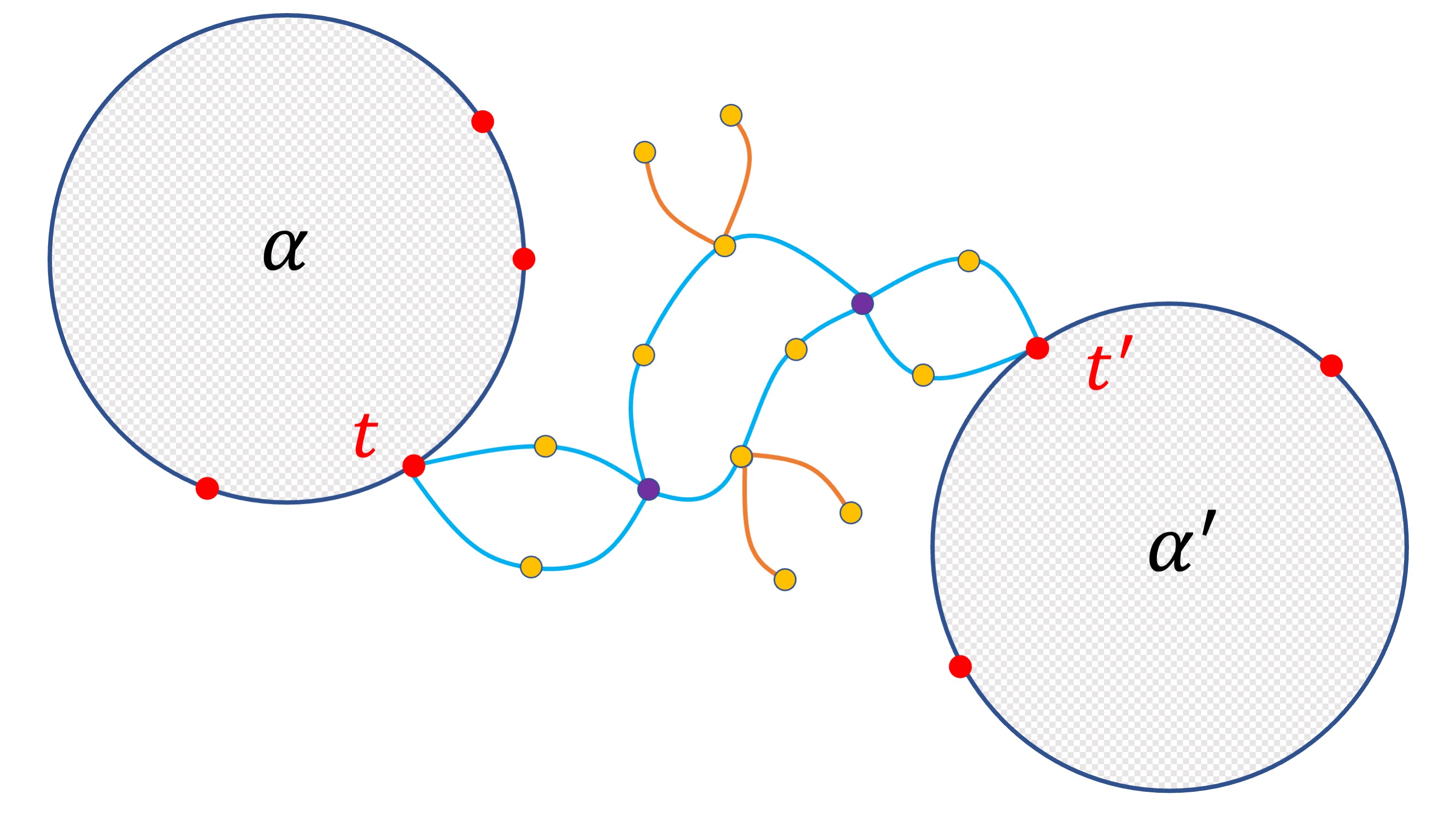}\label{fig: gluepath_2}}
	\caption{An illustration of cutting open and gluing an $f$-face instance along path $\Pi$.\label{fig: splitting h-hole}}
\end{figure}

It remains to complete the proof by induction. We first bound the size of graph $H$.
 
\paragraph{Size of $H$.} Since graph $H$ is obtained from graph $H'$ by gluing the portals on paths $\Pi_1,\Pi_2$ back to their original vertices, $|V(H)|\le |V(H')|$, so it suffices to bound the number of vertices in $H'$.
Recall that $T'$ is the set that contains all portals and terminals in $G'$, then
\[
\begin{split}
|V(H')| &
\le |T'|\cdot \bigg(\frac{cf\cdot \log |T'|}{\eps''}\bigg)^{c(f-1)^2}
\le \frac{c\cdot 2^f\cdot |T|}{\eps/2f^2}\cdot \bigg(\frac{cf\cdot \log (c\cdot 2^f\cdot |T|/\eps')}{\eps\cdot (1-1/f^2)}\bigg)^{c(f-1)^2}\\
&  \le |T|\cdot \bigg(\frac{cf}{\eps}\bigg)^{c(f-1)^2+f}\cdot \bigg(\frac{\log (cf^2\cdot 2^f\cdot |T|/\eps)}{(1-1/f^2)}\bigg)^{c(f-1)^2}\\
& \le |T|\cdot \bigg(\frac{cf}{\eps}\bigg)^{c(f-1)^2+f}\cdot \bigg(\log |T|+ \log (cf^2\cdot 2^f/\eps)\bigg)^{c(f-1)^2}\\
& \le |T|\cdot \bigg(\frac{cf^2 \log |T|}{\eps}\bigg)^{cf^2},
\end{split}
\]
where we have used the previously set parameters $\eps'=\eps/2f^2$ and $\eps''=(1-1/f^2)\cdot \eps$, and the fact that $(1-\frac{1}{f^2})^{-c(f-1)^2}\le e^{c}< c^{c}$, as $c$ is large enough.

\paragraph{Quality of $H$.}
Recall that we sliced $G$ open along $\Pi$ to obtain $G'$, computed a $(1+\eps'')$ pattern emulator $H'$ for $G'$, and then glue $H'$ back to obtain $H$.
For the purpose of analysis, image that we have also glued $G'$ back to obtain $G''$. The analysis is complete by the following claims.

\begin{claim}
\label{clm: cut glue emulator}
$(H,T)$ is a quality-$(1+\eps'')$ pattern emulator for $(G'',T)$.
\end{claim}
\begin{proof}
Consider any pair $t,t'$ of terminals. 
Let $P$ be some pattern shortest path from $t$ to $t'$ in $G''$.
If $P$ contains some portal in $Y$, then we let $y^1,\ldots,y^r$ be the portals appearing on $P$ in this order.
Decompose $P$ into $r+1$ subpaths: $P^0$ between $t$ and $y^1$; for each $1\le i\le r-1$, $P^i$ between $y^{i}$ and $y^{i+1}$, and $P_r$ between $y^{r}$ and $t'$.
Assume that for each $1\le i\le r-1$, the path $P^i$ starts on the $a_i$-side of the glued paths and ends on the $b_i$-side of the glued paths, where $a_i,b_i\in \set{1,2}$. Define $a(t), b(t')\in \set{1,2}$ similarly for terminals $t,t'$.
As portals are considered as terminals in instance $G'$ and $H'$ is a quality-$(1+\eps)$ pattern emulator for $G'$ with respect to them, there exist in $H'$ 
\begin{itemize}
\item a path $Q^0$ connecting the copy $t_{a(t)}$ of $t$ to the copy $y^1_{b_0}$ of $y^1$ with the same pattern as $P^0$ in $G'$, such that $w(Q^0)\le (1+\eps'')\cdot w(P^0)$;
\item for each $1\le i\le r-1$, a path $Q^i$ connecting the copy $y^{i-1}_{a_{i}}$ of $y^{i-1}$ to the copy $y^{i}_{b_i}$ of $y^i$ with the same pattern as $P^i$ in $G'$, such that $w(Q^i)\le (1+\eps'')\cdot w(P^i)$;
\item a path $Q^r$ connecting the copy $y^r_{a_r}$ of $y^r$ to the copy $t'_{b(t')}$ of $t'$ with the same pattern as $P^r$ in $G'$, such that $w(Q^r)\le (1+\eps'')\cdot w(P^r)$.
\end{itemize}

For each $1\le i\le r-1$, as the region surrounded by $P^i$ and $P$ contains no other face, and $Q^i$ is in the same pattern with $P^i$, the region surrounded by $Q^i$ and $P$ contains no other face. Therefore, if we concatenate all these paths $Q_0,\ldots,Q_r$ in $H$, we obtain a path $Q$ connecting $t$ to $t'$ in $H$ with the same pattern as $G''$. Moreover,
\[w(Q)=\sum_{0\le i\le r}w(Q^i)\le \sum_{0\le i\le r}(1+\eps'')\cdot w(P^i)=(1+\eps'')\cdot w(P).\]
The arguments for showing that for any pattern-shortest path $Q$ from $t$ to $t'$ in $H''$, there exists a path $P$ in $G''$ with $w(P)\le w(Q)$ that has the same pattern with $Q$ is symmetric. Altogether, the proof is complete.
\end{proof}

\begin{claim}
	\label{clm: cut glue}
	$(G'',T)$ is a quality-$(1+\eps')$ pattern emulator for $(G,T)$.
\end{claim}
\begin{proof}
Consider any pair $t,t'$ of terminals. 
Let $P$ be some pattern shortest path from $t$ to $t'$ in $G$. If path $P$ is vertex-disjoint from $\Pi$, then the same path exists in $G''$, and it is easy to verify that it has the same pattern as $P$ in $G$. We assume from now on that $P$ intersects $\Pi$, and assume without loss of generality that $P$ intersects $\Pi$ at a vertex $x$ and leaves from the $1$-side.
Let $\Phi$ be the pattern of the subpath of $P$ between $t$ and $x$.
	
From the property of pattern-respecting $\eps'$-cover, there exists a vertex $y\in Y$, such that 
$$\dist^{\Phi}(t,x)\le \dist^{\Phi}(t,y)+\dist(y,x)\le (1+\eps')\cdot \dist^{\Phi}(t,x).$$ 
Consider the path in $G''$ formed by the concatenation of (i) the $\Phi$-pattern shortest path connecting $t$ to $y$; (ii) the subpath of $P_2$ connecting $y$ to $x_1$; and (iii) the copy of the subpath of $P$ between $x_1$ and $t'$.
From the inequality above, such a path has total length at most $(1+\eps')\cdot w(P)$ and has the same pattern with $P$.

On the other hand, it is easy to observe that between every pair of terminals, any pattern shortest path only has greater length in $G''$. Altogether, the proof is complete.	
\end{proof}

From \Cref{clm: cut glue} and \Cref{clm: cut glue emulator}, $(H,T)$ is a $(1+\eps'')\cdot (1+\eps')\le (1+\eps)$ pattern emulator for $(G,T)$. This completes the proof by induction of \Cref{lem: O(1) face emulator}.

\subsection{Completing the proof of \Cref{main: upper}}

In this section, we complete the proof of \Cref{main: upper} using the results in previous subsections.

\paragraph{Structured $r$-divisions.} Let $G$ be a planar graph on $n$ vertices and let $T$ be its terminals. For any integer $r>1$, a \emph{structured $r$-division} of $G$ is a collection $\gset$ of subgraphs of $G$, such that
\begin{itemize}
\item $|\gset|=O(n/r)$;
\item the edge sets $\set{E(G')\mid G'\in \gset}$ partitions $E(G)$;
\item for each $G'\in \gset$, $|V(G')|\le r$, $|T\cap V(G')|=O(1+|T|\cdot r/n)$, and if we call vertices in $G'$ that appear in some other graph in $\gset$ its \emph{boundary vertices}, then
\begin{itemize}
\item the number of boundary vertices is $O(\sqrt{r})$;
\item in the planar drawing of $G'$ induced by the planar drawing of $G$, all boundary vertices lie on $O(1)$ faces.
\end{itemize}
\end{itemize}

We use the following lemma in \cite{chang2022near} for computing structured $r$-divisions.

\begin{lemma}[Lemma 5.5 in \cite{chang2022near}]
\label{lem: r-division}
Given any planar instance $(G,T)$ and any integer $r\ge 0$, we can efficiently compute a structured $r$-division.
\end{lemma}

We now prove \Cref{main: upper}. The algorithm is simply: compute a structured $r$-division $\gset$ of the input graph $G$, construct cut sparsifiers for graphs $G'\in \gset$ using \Cref{lem: O(1) face}, and then glue them together to obtain a cut sparsifier for $G$.
It turns out that we need to apply the algorithm several times in order to obtain a near-linear size cut sparsifier, each time with different parameters $r,\eps$.

\begin{lemma}
\label{lem: recursive}
Given any planar instance $(H,U)$ and any $0<\eps<1$, we can efficiently compute a quality-$(1+\eps)$ planar cut sparsifier $(\hat H, U)$ with size $|V(\hat H)|\le O\big(\sqrt{|V(H)| |U|}\cdot (\log |V(H)|/\eps)^{c'}\big)$, for some universal constant $c'>0$.
\end{lemma}

\begin{proof}
Denote $n=|V(H)|$ and $k=|U|$. Let $c'$ be a constant greater than the square of the product of the constant $c$ in \Cref{lem: O(1) face} and all hidden constants in the definition of structured $r$-divisions.

We first compute a structured $r$-division for $H$ with $r=n/k$ using the algorithm from \Cref{lem: r-division}, and obtain a collection $\hset$ of subgraphs of $H$, each being a $O(1)$-face instance. We then apply the algorithm from \Cref{lem: O(1) face} to each $H'\in \hset$ and obtain a $(1+\eps)$-cut sparsifier $\hat H'$, and finally glue all of them together to obtain $\hat H$.
From \Cref{lem: O(1) face} (here note that $f=O(1)$ and $f^2c\le c'$),
\[
|V(\hat H)|\le \sum_{H'\in \hset}|V(\hat H')|\le O\bigg(k\cdot \sqrt\frac{n}{k}\bigg)\cdot \bigg(\frac{cf \log n}{\eps}\bigg)^{cf^2}\le O\bigg(\sqrt{nk}\cdot \bigg(\frac{\log n}{\eps}\bigg)^{c'}\bigg).
\]
From \Cref{lem: divide}, the graph $\hat H$ is a quality-$(1+\eps)$ planar cut sparsifier. 
 This completes the proof.
\end{proof}

We now complete the proof of \Cref{main: upper} using \Cref{lem: recursive}. We start by applying the algorithm in \cite{krauthgamer2017refined} to obtain an exact cut sparsifier $(G_0,T)$ of $(G,T)$ with $|V(G)|=2^{2k}\cdot \poly(k)\le 2^{3k}$.
Then the algorithm consists of two phases.

In the first phase, we set parameters $L=\log k$ and $\eps'=\eps/6L$, and sequentially for each $i=1,\ldots,L$, we apply algorithm \Cref{lem: recursive} to graph $G_{i-1}$ to obtain graph $G_i$.
We show that $|V(G_L)|=(k/\eps)^{O(1)}$ (applying the next lemma with $i=L=\log k$).

\begin{claim}
For each $1\le i\le L$, $|V(G_i)|\le 2^{3k/2^i}\cdot (18k\log k/\eps)^{4c'}$.
\end{claim}
\begin{proof}
We prove by induction on $i$. The claim is true for $i=0$.
From \Cref{lem: recursive},
\[
\begin{split}
|V(G_i)| & \le \sqrt{|V(G_{i-1})|\cdot k}\cdot \bigg(\frac{\log (2^{3k})}{\eps'}\bigg)^{c'}
\le 
\sqrt{|V(G_{i-1})|}\cdot k\cdot\bigg(\frac{3k}{\eps/6\log k}\bigg)^{c'}\\
& \le \sqrt{2^{3k/2^i}\cdot (18\log k/\eps)^{4c'}}\cdot\bigg(\frac{3k}{\eps/6\log k}\bigg)^{2c'}\\
& \le 2^{3k/2^{i+1}}\cdot (18\log k/\eps)^{2c'}\cdot \bigg(\frac{3k}{\eps/6\log k}\bigg)^{2c'}
\le 2^{3k/2^i}\cdot (18k\log k/\eps)^{4c'}.
\end{split}
\]
\end{proof}

In the second phase, we set parameters $\bar L=2\log k\log k$ and $\bar\eps'=\eps/6L$, and sequentially for each $i=1,\ldots,\bar L$, we apply algorithm \Cref{lem: recursive} to graph $G_{L+i-1}$ to obtain graph $G_{L+i}$.

\begin{claim}
	For each $1\le i\le \bar L$, $|V(G_{L+i})|\le k^{1+5c'/2^i}\cdot (5c'\log k/\bar\eps')^{2c'}$.
\end{claim}
\begin{proof}
	We prove by induction on $i$. The claim is true for $i=0$.
	From \Cref{lem: recursive},
	\[
	\begin{split}
	|V(G_{L+i})| & \le \sqrt{|V(G_{L+i-1})|\cdot k}\cdot \bigg(\frac{\log (k^{5c'})}{\bar\eps'}\bigg)^{c'}
	\le 
	\sqrt{|V(G_{L+i-1})|}\cdot \sqrt{k}\cdot\bigg(\frac{5c'\log k}{\bar\eps'}\bigg)^{c'}\\
	& \le  k^{\frac{1}{2}\cdot (1+1+5c'/2^{i-1})}\cdot \sqrt{(5c'\log k/\bar\eps')^{2c'}}\cdot\bigg(\frac{5c'\log k}{\bar\eps'}\bigg)^{c'}\\
	& \le k^{1+5c'/2^i}\cdot (5c'\log k/\bar\eps')^{2c'}.
	\end{split}
	\]
\end{proof}
Applying $i=\bar L$, we get that $|V(G_{L+\bar L})|\le O(k\cdot \poly(\log k/\eps))$ (as $c'$ is a universal constant).

On the other hand, from \Cref{obs: chain}, we get that instance $(G_{L+\bar L},T)$ is a cut sparsifier of  instance $(G,T)$ with quality
$(1+\eps/6L)^L\cdot (1+\eps/6\bar L)^{\bar L}\le (1+\eps)$. The proof is complete.

\section{Proof of \Cref{quasi_1}}
\label{sec: quasi_exact}

In this section, we provide the proof of \Cref{quasi_1}. 

We first prove the upper bound that every quasi-bipartite graph with $k$ terminals admits an exact contraction-based cut sparsifier on $k^{O(k^2)}$ vertices. 
Let $G$ be a quasi-bipartite graph with a set $T$ of $k$ terminals.
Throughout, we will assume that for every $\emptyset \subsetneq S\subsetneq T$, there is a unique min-cut in $G$ separating $S$ from $T\setminus S$. This can be achieved by slightly perturbing the edge weights $\set{w(e)}_{e\in E(G)}$.

We use the definition of \emph{profiles} studied in previous works \cite{hagerup1998characterizing,khan2014mimicking}.
For a vertex $v\in V(G)$, its profile $\pi^v$ is defined to be a $(2^{|T|}-2)$-dimensional vector whose coordinates are indexed by proper subsets $S$ of $T$. Specifically, for every $\emptyset \subsetneq S\subsetneq T$, $\pi^v_S=1$ iff $v$ lies on the side of $S$ in the $(S,T \setminus S)$ min-cut in $G$, otherwise $\pi^v_S=0$. 
Let $\Pi(G)$ be the collection of all distinct profiles of the vertices in $G$.
The following result was proved in \cite{hagerup1998characterizing}.

\begin{lemma}
Every graph $G$ admits a quality-$1$ contraction-based cut sparsifier of size $|\Pi(G)|$.
\end{lemma}

Therefore, to prove \Cref{quasi_1}, it is sufficient to prove the following lemma.

\begin{lemma} \label{quasi_profile}
For any quasi-bipartite graph $G$ with $k$ terminals, $|\Pi(G)|=k^{O(k^2)}$.
\end{lemma}

To prove \Cref{quasi_profile}, we will show that $\Pi(G)$, when viewed as a family of sets (instead of vectors), has bounded VC-dimension.
Specifically, consider the ground set $\uset=\set{S\mid \emptyset \subsetneq S\subsetneq T}$.
Now each profile $\pi^v$ naturally defines a subset of the ground set $\uset$: $\Pi^v=\set{S\mid \pi^v_S=1}$.
Abusing the notation, we also write $\Pi(G)=\set{\Pi^v\mid v\in V(G)}$.
We will show that the set family $\Pi(G)$ on the ground set $\uset$ has VC-dimension $O(k \log k)$. Then from Sauer-Shelah lemma \cite{shelah1972combinatorial,sauer1972density},  
$$|\Pi(G)|\le |\uset|^{O(k \log k)}\le (2^k)^{O(k \log k)} = k^{O(k^2)}.$$  

\paragraph{Shattering sets and VC dimension.} Let $\fset$ be a set family on the ground set $\uset$.
Let $U$ be a subset of $\uset$. We say that the family $\fset$ \emph{shatters} $U$, iff for every subset $U'\subseteq U$, there exists some set $F\in \fset$, such that $F\cap U=U'$. The VC-dimension of the family $\fset$ is the maximum size of a set $U$ shattered by $\fset$.

The remainder of this section is dedicated to the proof of the following claim.

\begin{claim}
\label{clm: VC-dim}
The VC-dimension of $\Pi(G)$ on the ground set $\uset$ is $O(k \log k)$.
\end{claim}
\begin{proof}
For every subset $S \subseteq T$, we denote by $w_v(S)$ the total weight of all the edges between $v$ and the terminals in $S$. Since $G$ is a quasi-bipartite graph, for every $S$, the value of $\pi^v_S$ only depends on $w_v(S)$ and $w_v(T\setminus S)$. Specifically, for every $v$, $\pi^v_S=1$ iff $w_v(S)>w_v(T\setminus S)$.

Recall that the ground set $\uset$ contains all proper subsets of $T$ as its elements. Here we consider the subsets of $\uset$. We say that two subsets $\uset_1,\uset_2\subseteq \uset$ are \emph{similar}, iff
\begin{itemize}
\item $\uset_1,\uset_2$ are disjoint subsets of $\uset$;
\item $|\uset_1|=|\uset_2|$; and
\item for every terminal $t\in T$, $\big|\set{S\in \uset_1\mid t\in S}\big|=\big|\set{S\in \uset_2\mid t\in S}\big|$.
\end{itemize}
The proof of \Cref{clm: VC-dim} is completed by the following two observations.

\begin{observation}
Let $\uset_1,\uset_2$ be a similar pair. Then any subset $\uset'\subseteq \uset$ that contains both $\uset_1$ and $\uset_2$ is not shattered by $\Pi(G)$.
\end{observation}
\begin{proof}
As $\uset_1$ and $\uset_2$ are similar, for every non-terminal $v\in V(G)$, $\sum_{S\in \uset_1} w_{v}(S) = \sum_{S \in \uset_2} w_{v}(S)$ and $\sum_{S \in \uset_1} w_v(T \setminus S) = \sum_{S \in \uset_2} w_v(T \setminus S)$ must hold. 
If there exists a set $\Pi^v\in \Pi(G)$ with $\Pi^v\cap \uset'=\uset_1$, then by definition,
$\pi^v_S=1$ for all $S \in \mathcal{U}_1$ and $\pi^v_S=0$ for all $S \in \mathcal{U}_2$. 
However, this means that
\[
\sum_{S\in \uset_1} w_{v}(S)> |\uset_1|\cdot \frac{w_{v}(T)}{2}= |\uset_2|\cdot \frac{w_{v}(T)}{2}
> \sum_{S\in \uset_2} w_{v}(S) = \sum_{S\in \uset_1} w_{v}(S),
\]
a contradiction.
\end{proof}

\begin{observation}
Every subset $\uset'\subseteq \uset$ with $|\uset'|>2k\log k$ admits a similar pair $\uset_1,\uset_2\subseteq \uset'$.
\end{observation}
\begin{proof}
Let $\uset'$ be a subset of $\uset$ with $|\uset'|>k\log k$.
First, there are $\binom{|\uset'|}{\floor{|\uset'|/2}}$ size-$\floor{|\uset'|/2}$ subsets of $\uset'$. Second, for such each such subset $\uset''$, if we define the $T$-dimensional vector $\alpha[\uset'']=(\alpha[\uset'']_t)_{t\in T}$ as
$\alpha[\uset'']_t=\big|\set{S\in \uset''\mid t\in S}\big|$, then there are a total of $(\floor{|\uset'|/2}+1)^k$ possible vectors $\alpha[\uset'']$.

This impies that, when $|\uset'|>2k\log k$, $\binom{|\uset'|}{\floor{|\uset'|/2}}>(\floor{|\uset'|/2}+1)^k$, and so there must exist a pair $\mathcal{U}'_1,\mathcal{U}'_2$  of distinct subsets of $\mathcal{U}'$, such that $\card{\mathcal{U}_1} = \card{\mathcal{U}_2} = \floor{|\uset'|/2}$ and the vectors $\alpha[\uset_1],\alpha[\uset_2]$ are identical.
It is then easy to verify that $\uset_1\setminus \uset_2,\uset_2\setminus \uset_1$ are a similar pair, both contained in $\uset'$. 
\end{proof}
\end{proof}

We now prove the lower bound that for every integer $k$, there exists a quasi-bipartite graph with $k$ terminals whose exact contraction-based cut sparsifier must contain $\Omega(2^k)$ vertices.

We construct a graph $G$ as follows. The terminal set contains $k$ vertices and is denoted by $T$. For every subset $\emptyset\subsetneq S \subsetneq T$ with even size, we add a new vertex $v_S$ and connect it to every vertex of $S$.
Every edge is given a capacity that is chosen uniformly at random from the interval $(1-2^{k},1+2^{-k})$. In this way we can guarantee that with high probability, sets $E',E''$ of edges of $G$ have the same total weight iff $E'=E''$.
And clearly, $G$ is a quasi-bipartite graph with $|V(G)|=\Omega(2^k)$.

The proof of the $\Omega(2^k)$ size lower bound is completed by the following claims.

\begin{claim}
\label{clm: profiles}
Every pair of distinct vertices have different profiles.
\end{claim}
\begin{proof}
Consider a pair $v=v_S,v'=v_{S'}$ of vertices in $G$.
As $G$ is a quasi-bipartite graph, the profile of $v$ only depends on incident edges of $v$.
We will show that there always exists a subset $U\subseteq T$ of terminals, such that $v,v'$ lie on different sides of the $(U,T\setminus U)$ min-cut in $G$.

Recall that $S,S'$ are distinct subsets of $T$ with even size. Assume w.l.o.g. that $S\setminus S'\ne \emptyset$. If $|S\cap S'|$ is even, then we let $U$ contains exactly half of the elements in $S\cap S'$. Now if $v$ lies on the $U$ side of the $(U,T\setminus U)$ min-cut in $G$, then we are done, as $v'$ lies on the $T\setminus U$ size of the min-cut, then we are done we add to $U$. Otherwise, we add all elements in $S\setminus S'$, ensuring that $v$ lies on the $T\setminus U$ side of the $(U,T\setminus U)$ min-cut while $v'$ lies on the $U$ side.
Similarly, if $|S\cap S'|$ is even, then we let $U$ contains $\frac{|S\cap S'|+1}{2}$ elements in $S\cap S'$, and we can adjust $U$ similarly (by optionally adding to it all elements of $S\setminus S'$) to ensure that $v,v'$ lie on different sides of the $(U,T\setminus U)$ min-cut in $G$.
\end{proof}

\begin{claim}
No pair of vertices can be contracted in an exact contraction-based cut sparsifier of $G$.
\end{claim}
\begin{proof}
Assume for contradiction that vertices $v_S,v_{S'}$ are contracted. We have shown in \Cref{clm: profiles} that there exists a subset $U$ of $T$, such that vertices $v_S,v_{S'}$ lie on different sides of the $(U,T\setminus U)$ min-cut in $G$.
Therefore, if $v_S,v_{S'}$ are merged in the sparsifier, then the set $E''$ of edges constituting the $(U,T\setminus U)$ min-cut in the cut sparsifier $H$ must be different from the set $E'$ of edges constituting the $(U,T\setminus U)$ min-cut in $G$, and therefore they must have different total weight (meaning that $H$ cannot be an exact cut sparsifier for $G$), since we have assumed that sets $E',E''$ of edges of $G$ have the same total weight iff $E'=E''$.
\end{proof}

\section{Proof of \Cref{quasi_apx}}
\label{sec: quasi_apx}

In this section, we provide the proof of \Cref{quasi_apx}, showing that every quasi-bipartite graph with $k$ terminals admits a quality-$(1+\eps)$ contraction-based cut sparsifier of size $k^{O(1/\eps^2)}\cdot (1/\eps)^{O(1/\eps^4)}$. 


\subsection{Sparsifying stars by sampling}
\label{subsec: sparsified star}

Throughout, we use the parameter $c = \ceil{100/\eps^2}$. 

For every non-terminal $v$ in $G$, denote by $G_v$ the star induced by all incident edges of $v$. We will construct another star $H_v$ that contains at most $c$ edges to mimick the behaviors of $G_v$, as follows.

Denote by $w$ the total weight of all edges in $G_v$.
We say an edge $e$ is \emph{heavy} iff $w(e)\ge w/c$, otherwise we say it is \emph{light}.
Denote by $h$ the number of heavy edges in $G_v$ and by $w^h$ the total weight of them, so $h\le c$.
We define a distribution $\dset$ on light edges as follows: the probability for a light edge $e$ is 
$\text{Pr}_{\dset}(e)=\frac{w(e)}{w-w^h}.$
If $|V(G_v)|\le c$, then we let $H_v=G_v$.
Otherwise, we sample $c-h$ light edges from $\dset$ with replacement. For each sampled edge, we assign to it a new weight of $\frac{w-w^h}{c-h}$.
The graph $H_v$ simply contains all heavy edges in $G_v$ and all sampled light edges (with their new weight).

For a subset $S\subseteq T$, let $w_S(v)$ be the total weight of edges in $G_v$ connecting $v$ to a vertex in $S$, so $w=w_T(v)$, and similarly let $w'_S(v)$ be the total weight of edges in $H_v$ connecting $v$ to a vertex in $S$. 
In this subsection we will omit $v$ and only write $w_S,w'_S$ for notational simplicity.
The following observation is immediate.

\begin{observation}
\label{obs: property}
Graph $H_v$ contains at most $c$ edges with a total weight of $w$. Moreover, for every subset $S\subseteq T$, $\ex{w'_S}=w_S$.
\end{observation}

We now measure the similarity of graphs $H_v$ and $G_v$ in terms of terminal cuts.
For a subset $S\subseteq T$, if $w'_S>w/2$, then in graph $H_v$, $v$ lies on the $S$ side of the $(S, T\setminus S)$ min-cut, otherwise $v$ lies on the $T\setminus S$ side. 
We prove the following claim.

\begin{claim}
\label{clm: difference in contribution}
For a subset $S\subseteq T$ with $w_S<w/2$, $\big(w-2\cdot w_S\big) \cdot \pr{w'_S>w/2} \le \eps \cdot w_S$.
\end{claim}
\begin{proof}
Assume by scaling that $w=1$.
We use the following Chernoff bound on negative associated random variables.
	
    \begin{lemma} [\cite{dubhashi1996balls}] \label{chernoff}
		Suppose $X_1,\dots,X_n$ are negatively associated random variables taking values in $\{0,1\}$. Let $X=\sum_{1\le  i\le n}X_i$ and $\mu=\ex{X}$. Then for any $\delta>0$,
		\begin{itemize}
			\item $\pr{X > (1+\delta)\mu} \le e^{-\frac{\delta^2 \mu}{(2+\delta)}}$; and
			\item $\pr{X > (1+\delta)\mu} \le \frac{e^{-\mu}}{(1+\delta)^{1+\delta}}$.
		\end{itemize}
	\end{lemma}
Let $X_e$ be the random variable for the new weight of edge $e$ in $H_v$ (that is, if $e$ is heavy, then $X_e=w(e)$, if $e$ is light and sampled, $X_e=\frac{w-w^h}{c-h}$, and if $e$ is light and not sampled, $X_e=0$). According to our sampling process, variables $\set{X_e}_{e\in E(G_v)}$ are negatively associated.
From \Cref{chernoff}, letting $\alpha = 1/2 - w_S$ and $\delta = \alpha/w_S$, we have
	\begin{itemize}
		\item $\pr{w'_S > 1/2} \le e^{-\frac{c \delta^2 w_S}{(2+\delta)}} = e^{-\frac{c \delta \alpha }{(2+\delta)}}$; and
		\item $\pr{w'_S > 1/2} \le \frac{e^{-c w_S}}{(1+\delta)^{1+\delta}} \le \frac{e^{-c w_S}}{\big(\frac{1}{2w_S}\big)^{\frac{1}{2w_S}}}$;
	\end{itemize}
(we have assumed that all edges are light, as the heavy edges have weights unchanged and can only help the concentration property of $w'_S$).
We consider the following three cases. 
	\paragraph{Case 1: $1/2 - \eps/5 < w_S \le 1/2$.} In this case, $(1-2\cdot w_S)<2\cdot\eps/5<\eps\cdot w_S$.
	\paragraph{Case 2: $1/4 < w_S < 1/2 - \eps/5$.} In this case, since $\delta < 1$,
	\begin{align*}
	\pr{w'_S > 1/2} \le e^{-\frac{c \delta \alpha}{3}} < e^{- \frac{2 c \alpha^2}{3}} < \frac{3}{2 c \alpha^2} < \frac{\eps}{12\alpha},
	\end{align*}
	where the second inequality is because $w_S < 1/2$; the third one is becuase for any $x>0$, $e^{-x}<1/x$; and the last one is because $c=\ceil{100/\eps^2}$ and $\alpha>\eps/5$. Therefore, $2\alpha \cdot \pr{w'_S > 1/2} < \eps/6 < \eps\cdot w_S$. 
	\paragraph{Case 3: $w_S \le 1/4$.} In this case, $2\cdot w_S \le 1/2$, so $\pr{w'_S > 1/2} \le \frac{e^{-c w_S}}{(1/2w_S)^2} \le \frac{4w_S}{c}$. Therefore, $2\alpha \cdot \pr{w'_S > 1/2} < \frac{12 w_S}{c} < \eps\cdot w_S$, since $\alpha < 1/2$.
\end{proof}

The following statement is a byproduct of the Case 3 above and will be useful later.
\begin{claim} \label{prop:v2}
	If $w_S \le 1/(2e)$, $\pr{w'_S > 1/2} \le e^{-\frac{w}{2w_S}}$.
\end{claim}

\subsection{Construction of the cut sparsifier}

\subsubsection*{Step 1. Computing special cuts and keeping important vertices}

For every pair $t,t'$ of terminals in $T$, we compute a min-cut in $G$ separating $t$ and $t'$. This min-cut induces a terminal partition $(S_{t,t'},T\setminus S_{t,t'})$ where $t\in S_{t,t'}$ and $t'\in T\setminus S_{t,t'}$. We call such partitions \emph{special cuts}.

\begin{claim} \label{lem:contribution}
For every vertex $v$ and every subset $S\subseteq T$, there is a special terminal cut $(\tilde S,T\setminus \tilde S)$, such that 
$$\frac{\min \{w_{\tilde S}(v),w(v)-w_{\tilde S}(v)\}}{\mc_G(\tilde S)} \ge \frac{1}{k}\cdot \frac{\min \{w_S(v),w(v)-w_S(v)\}}{\mc_G(S)}.$$
\end{claim}

\begin{proof}
Denote $\alpha=\frac{\min \{w_{S}(v),w(v)-w_{S}(v)\}}{\mc_G(S)}$, so $\frac{w_S(v)}{\mc_G(S)},\frac{w(v)-w_S(v)}{\mc_G(S)}\ge \alpha$. Therefore, there exist terminals $t\in S$ and $t' \notin S$ such that both $(v,t),(v,t')$ have weight at least $\frac{\alpha \cdot \mc_G(S)}{k}$. Consider the special cut $(S_{t,t'},T\setminus S_{t,t'})$ which is the min-cut in $G$ separating terminals $t,t'$, so $\mc_G(S_{t,t'})\le \mc_G(S)$.
Denote $\tilde S=S_{t,t'}$.
On the other hand, since the cut $(\tilde S,T\setminus \tilde S)$ separates $t,t'$, $(v,t),(v,t')$ cannot both lie in $E(v,\tilde S)$ or $E(v,T\setminus \tilde S)$. Therefore, $w_{\tilde S}(v),w(v)-w_{\tilde S}(v)\ge \frac{\alpha \cdot \mc_G(S)}{k}$.
Altogether,
\[\frac{\min\set{w_{\tilde S}(v),w(v)-w_{\tilde S}(v)}}{\mc_G(\tilde S)}\ge \frac{\alpha \cdot \mc_G(S)}{k\cdot \mc_G(\tilde S)}=\frac{\alpha}{k}=\frac{1}{k}\cdot \frac{\min \{w_S(v),w(v)-w_S(v)\}}{\mc_G(S)}.\]
\end{proof}

Let $\eta = \frac{\eps^4}{1000k}$. We call a vertex $v$ \emph{important} if there is a special terminal cut $(S,T\setminus S)$ such that the contribution of $v$ on it (that is, the value of $\mc_{G_v}(S,T\setminus S)$) is at least $(\eta/k)\cdot \mc_G(S)$. There are at most $k \cdot k/\eta  = O(k^3/\eps^4)$ important vertices. We keep all important vertices and their incident edges in the cut sparsifier. That is, important vertices are not contracted with any other vertices.

According to \Cref{lem:contribution}, every non-important vertex $v$ contributes at most $\eta \cdot \mc_G(S)$ to any cut $\mc(S,T\setminus S)$. We handle them next.

\subsubsection*{Step 2. Contracting non-important vertices based on their sparsified stars}

Recall that for each non-important non-terminal $v$, we have constructed in \Cref{subsec: sparsified star} a sparsified star $H_v$ mimicking the original star $G_v$ in $G$. Let $\pi_v$ be the profile of vertex $v$ in graph $H_v$. For each profile $\pi$, we contract all non-important vertices $v$ with $\pi_v=\pi$ into a single node. This completes the construction of the contraction-based sparsifier. Denote the resulting graph by $G'$.

Clearly, $G'$ is obtained from $G$ by contractions only, so it is indeed a contraction-based sparsifier. The next observation shows that it contains at most 
$k^{O(1/\eps^2)}\cdot (1/\eps)^{O(1/\eps^4)}$ vertices.

\begin{observation}
$|V(G')|\le k^{O(1/\eps^2)}\cdot (1/\eps)^{O(1/\eps^4)}$.
\end{observation} 
\begin{proof}
We have shown that the number of important vertices is $O(k^3/\eps^4)$. From \Cref{quasi_1}, the number of different profiles of non-important vertices is at most $\binom{k}{c}\cdot c^{O(c^2)}\le k^{O(1/\eps^2)}\cdot (1/\eps)^{O(1/\eps^4)}$.
\end{proof}

In the rest of this section, we prove that $G'$ is indeed a quality-$(1+O(\eps))$ cut sparsifier.

Fix a subset $S\subseteq T$ and we analyze the difference between $\mc_G(S)$ and $\mc_{G'}(S)$.
For a non-terminal $v$, let $x_v$ and $x'_v$ be the contribution of the vertex $v$ for the cut $S$ in $G$ and $G'$ respectively, so $\mc_G(S) = \sum_v x_v$. 

\begin{observation}
If $v$ is an important vertex, then $x'_v=x_v$. If $v$ is a non-important vertex, then $x'_v$ is either $x_v$ or $w(v)-x_v$. Moreover, $\ex{x'_v} \le (1+\eps)x_v$ and $x_v \le \eta \cdot \mc_G(S)$.
\end{observation}
\begin{proof}
Since we have kept all important vertices in Step 1, their contribution to every terminal min-cut does not change.
For a non-important vertex $v$, since we constructed a sparsified star $H_v$ for $G_v$, and decide the side of $v$ for each terminal cut using the profile of $v$ in $H_v$ instead of $G_v$. 
Therefore, when $w_S(v)\le w(v)/2$ and $w'_S(v)\le w(v)/2$, then $x'_v=x_v$, otherwise $x'_v=w(v)-x_v$, and the difference is $x'_v-x_v=w(v)-2x_v$.
We have proved in \Cref{clm: difference in contribution} that for a subset $S\subseteq T$ with $w_S(v)<w(v)/2$, $\big(w(v)-2\cdot w_S(v)\big) \cdot \pr{w'_S(v)>w(v)/2} \le \eps \cdot w_S(v)$, and so $\ex{x'_v} \le (1+\eps)x_v$. Finally, from our discussion in Step 1, $x_v \le \eta \cdot \mc_G(S)$. 
\end{proof}

We classify all non-important vertices into two types. 
Let $V_1$ be the set that contains all vertices $v$ such that $x_v/w(v) \ge \eps^2/k^2$, and let $V_2$ contain the rest of non-important vertices. 
The proof is completed by the following claims.

\begin{claim}
$\pr{\sum_{v \in V_1} (x'_v - \ex{x'_v}) \le \eps\cdot \mc_G(S)} \le 2^{-2k}$.
\end{claim}
\begin{proof}
We use the following version of Hoeffding's inequality.
\begin{lemma} \label{hoeffding}
	Let $X_1,\dots,X_n$ be independent random variables such that $a_i \le X_i \le b_i$. Let $S_n = \sum_{1\le  i\le n}X_i$. Then for any $t>0$,
	$\pr{S_n - \ex{S_n} \ge t} \le e^{-\frac{2t^2}{\sum_{i=1}^n(b_i-a_i)^2}}.
	$
\end{lemma}

Let $\gamma = \frac{\eps}{10}$.
For any vertex $v$ in $V_1$, $x'_v$ is bounded by $x_v/\gamma$, so we can upper and lower bound $x'_v$ by $b_v = x_v/\gamma$ and $a_v=0$. Also note that 
$\sum_{v\in V_1} (b_v-a_v)^2 \le \frac{1}{\gamma^2} \sum_{v\in V_1} x^2_v.$
Since for any $v \in V_1$, $x_v \le \eta \cdot \mc_G(S)$, $\sum_{v\in V_1} x^2_v \le \eta \cdot(\mc_G(S))^2$, which means 
$$
\sum_{v\in V_1} (b_v-a_v)^2 \le \frac{\eta}{\gamma^2} \cdot (w(v))^2 \le \frac{\eps^2}{10k}\cdot (\mc_G(S))^2.
$$
From \Cref{hoeffding}, 
$$
\pr{\sum_{v \in V_1} (x'_v - \ex{x'_v}) \le \eps\cdot \mc_G(S)} \le e^{-\frac{2\eps^2 (\mc_G(S))^2}{\frac{\eps^2}{10k}\cdot(\mc_G(S))^2}} \le 2^{-2k}.
$$
\end{proof}

\begin{claim}
$\pr{\sum_{v\in V_2} x'_v > \sum_{v\in V_2} x_v + 2\eps\cdot \mc_G(S)} <e^{-2k}$.
\end{claim}
\begin{proof}
We use the approach for proving the Chernoff bound, i.e. applying Markov's inequality on the moment generate function of $\sum_{v \in V_2} x'_v$. Since for every $v \in V_2$, $x_v \le \gamma \cdot w(v)$, from \Cref{prop:v2}, $\pr{x'_v = w(v)-x_v} < e^{-\frac{w(v)}{2x_v}}$. For any $t>0$, 
\begin{align*}
\ex{e^{tx'v}} < e^{tx_v} + e^{-\frac{w(v)}{2x_v}}(e^{tw(v)}-e^{tx_v}) < e^{tx_v}(1+e^{-\frac{w(v)}{2x_v}}(e^{tw(v)}-1)) < e^{tx_v + e^{-\frac{w(v)}{2x_v}}(e^{tw(v)}-1)}.
\end{align*}

Let $t = \frac{2k}{\eps \cdot\mc_G(S)}$. Since $x_v \le \eta \cdot\mc_G(S)$, $t \le \frac{2 \eta k}{\eps x_v}$, which means $tw(v) \le \frac{2\eta k w(v)}{\eps x_v} \le \frac{w(v)}{4x_v}$. This means $e^{-\frac{w(v)}{2x_v}}(e^{tw(v)}-1)< e^{-\frac{w(v)}{4x_v}}$. Since $\frac{w(v)}{x_v} \ge (1/\gamma)$ and $\gamma=\frac{\eps}{10}$, we have $\frac{w(v)}{x_v}>4( \log (w(v))-\log(x_v) + \log (1/\eps))$. So $e^{-\frac{w(v)}{4x_v}}<\frac{\eps x_v}{w(v)}$. Therefore, if $tw(v) > 1$,
$$
\ex{e^{tx'_v}}<e^{tx_v + e^{-\frac{w(v)}{2x_v}}(e^{tw(v)}-1)}<e^{tx_v + \eps x_v/w(v)}<e^{(1+\eps)tx_v}.
$$
On the other hand, if $tw(v) \le 1$, $e^{tw(v)-1} < 2tw(v)$. Since $\frac{w(v)}{x_v} \ge (1/\gamma)$ and $\gamma=\frac{\eps}{10}$, we have $\frac{w(v)}{2x_v} > \log w(v) - \log x_v + \log (1/\eps)$, which means $e^{-\frac{w(v)}{2x_v}}<\eps\cdot \frac{x_v}{w(v)}$. Therefore,
$$
\ex{e^{tx'_v}}<e^{tx_v + e^{-\frac{w(v)}{2x_v}}(e^{tw(v)}-1)}<e^{tx_v + \eps\frac{x_v}{w(v)} \cdot tw(v)} < e^{(1+\eps)tx_v}.
$$

Therefore in either case, we have $\ex{e^{tx'_v}}<e^{(1+\eps)tx_v}$, and so
$$\ex{e^{t\sum_{v\in V_2} x'_v}} < e^{(1+\eps)t\sum_{v\in V_2} x_v}.$$ 
Therefore,
\[
\begin{split}
\pr{\sum_{v\in V_2} x'_v > \sum_{v\in V_2} x_v + 2\eps \cdot\mc_G(S)} & < e^{(1+\eps)t\sum_{v\in V_2} x_v-t(\sum_{v \in V_2} x_v + 2\eps \cdot\mc_G(S))}\\
& < e^{-\eps t \cdot\mc_G(S)}<e^{-2k}.
\end{split}
\]
\end{proof}

Therefore, $\pr{\sum_v x'_v > (1+3\eps)\cdot\mc_G(S)} < 2^{-2k+1}$. Taking the union bound over all subsets $S\subseteq T$, with probability at least $1-2^{k-1}$, $G'$ is indeed a quality-$(1+3\eps)$ cut sparsifier.

\section{Proof of \Cref{main: lower}}
\label{sec: lower}

In this section we prove \Cref{main: lower}, showing that for every $k\ge 1$ and $\eps>0$, there exists a quasi-bipartite graph with $k$ terminals, whose quality-$(1+\eps)$ contraction-based cut sparsifier must contain $k^{\Omega\big(\frac{1}{\eps \log(1/\eps)}\big)}$ vertices. We will make use of the following intermediate problem, which is a variant of the $0$-Extension problem.

\newcommand{\bhc}{\textnormal{\textsf{BHC}}}

\paragraph{Boolean Hypercube Contraction ($\bhc$).}
In an instance of $\bhc$, the input consists of
\begin{itemize}
	\item a graph $G=(V,E)$ with $V=\set{0,1}^d$, namely each vertex is a $d$-dimensional $0/1$-string;
	\item a set $T\subseteq V$ of $k$ terminals, such that every edge is incident to some terminal.
\end{itemize}
A solution is a mapping $f:V\to V$, such that for each terminal $t\in T$, $f(t)=t$. For solution $f$, its
\begin{itemize}
\item \emph{cost} is defined as $\vol(f)=\sum_{(u,v)\in E}||f(u)-f(v)||_1$; and
\item \emph{size} is defined as $|\set{f(v)\mid v\in V}|$, namely the size of its image set.
\end{itemize}

The goal is to compute a solution $f$ with small size and cost. Specifically, the following ratio called \emph{stretch}, is to be minimized:
\[\rho=\frac{\sum_{(u,v)\in E}||f(u)-f(v)||_1}{\sum_{(u,v)\in E}||u-v||_1}.\]

In fact, the $\bhc$ problem is the special case of the $0$-Extension problem first introduced by Karzanov \cite{karzanov1998minimum}, as a generalization of the multi-way cut problem and a special case of the metric labeling problem.
The $0$-Extension problem and some of its variants were shown to be concretely connected to cut/flow sparsifiers \cite{moitra2009approximation,leighton2010extensions,andoni2014towards,chen2024lower,chen20241+}. Specifically, $\bhc$ focuses on searching for ``canonical solutions'' of the \emph{$0$-Extension with Steiner Nodes} problem studied \cite{chen2024lower}, and moreover, instead of allowing the input graph to be any general graph, $\bhc$ only consider instances on the boolean hypercube.

The proof of \Cref{main: lower} consists of the following two lemmas.

\begin{lemma}
If for every $k\ge 1$ and $q\ge 1$, every $k$-terminal quasi-bipartite graph admits a quality-$q$ contraction-based cut sparsifier of size $s(k,q)$, then every $k$-terminal quasi-bipartite instance of $\bhc$ admits a solution of stretch $q$ and size $s(k,q)$.
\end{lemma}

\begin{proof}
Let $(G,T)$ be an instance of $\bhc$. Assume that it admits a quality-$q$ contraction-based cut sparsifier $H$ induced by a partition $\lset$ of $V$, so $H$ is obtained from $G$ by contracting each set $L\in \lset$ into a supernote $u_L$.
We use $H$ to construct a $\bhc$ solution of this instance as follows.

Recall that every vertex in $V(G)$ is a $d$-dimensional $0/1$ string. For each index $1\le i\le d$, define $T^i_0$ as the set of terminals $t$ whose $i$-th coordinate is $0$, and define $T^i_1$ as the set of terminals $t$ whose $i$-th coordinate is $1$, so $(T^i_0,T^i_1)$ is a partition of $T$. For a node $u_L$, we map it to a $d$-dimensional $0/1$ string $g(u_L)$ as follows. The $i$-th bit of $g(u_L)$ is $0$ iff $u_L$ lies on the side of $T^i_0$ in the min-cut in $H$ separating $T^i_0$ from $T^i_1$, and the $i$-th bit of $g(u_L)$ is $1$ iff $u_L$ lies on the side of $T^i_1$ in this min-cut.
Now for every vertex $v\in V(G)$, if it belongs to a set $L(v)\in \lset$, then we map it to $g(u_{L(v)})$, namely $f(v)=g(u_{L(v)})$. This completes the construction of $f$. Clearly $f$ is a mapping from $V$ to $V$, and the size of $f$ is exactly the size of $H$.

We now show that $f$ has stretch at most $q$.
For each $1\le i\le d$, we define $E_i$ as the set of all edges $(v,v')$ in $G$ where the $i$-th coordinate of $v$ is $0$ and the $i$-th coordinate of $v'$ is $1$. Clearly, $E_i$ is a cut in $G$ separating $T^i_0$ from $T^i_1$, so
\[
\sum_{1\le i\le d}\mc_G(T^i_0,T^i_1)\le \sum_{1\le i\le d}|E_i|\le \sum_{(v,v')\in E(G)}\sum_{1\le  i\le d}\mathbf{1}[(v,v')\in E_i]=\sum_{(v,v')\in E(G)}||v-v'||_1.
\]
On the other hand, from our construction of $f$, for each $1\le i\le d$, the min-cut in $H$ separating $T^i_0$ from $T^i_1$ is exactly the set $E'_i$ of edges connecting a supernode $u_L$ whose mapped string $g(u_L)$ has $0$ at its $i$-th coordinate to a supernode $u_L$ whose mapped string $g(u_L)$ has $1$ at its $i$-th coordinate. Therefore, 
\[
\sum_{1\le i\le d}\mc_H(T^i_0,T^i_1)= \sum_{1\le i\le d}|E'_i|= \sum_{(v,v')\in E}\sum_{1\le  i\le d}\mathbf{1}[(f(v),f(v'))\in E'_i]=\sum_{(v,v')\in E}||f(v)-f(v')||_1.
\]
Since $H$ is a quality-$q$ cut sparsifier of $G$, the stretch of function $f$ is at most
\[
\frac{\sum_{(v,v')\in E(G)}||f(v)-f(v')||_1}{\sum_{(v,v')\in E(G)}||v-v'||_1}\le \frac{\sum_{1\le i\le d}\mc_H(T^i_0,T^i_1)}{\sum_{1\le i\le d}\mc_G(T^i_0,T^i_1)}\le \frac{\sum_{1\le i\le d} \mc_G(T^i_0,T^i_1)\cdot q}{\sum_{1\le i\le d}\mc_G(T^i_0,T^i_1)}=q.
\]
\end{proof}

\begin{lemma}
For every $k\ge 1$ and $\eps>0$, there exists a $k$-terminal quasi-bipartite instance of $\bhc$, such that any solution of stretch $(1+\eps)$ and must have size at least $k^{\Omega(1/(\eps\log (1/\eps)))}$.
\end{lemma}

\begin{proof}
We construct the hard instance $(G,T)$ of $\bhc$  as follows. The vertex set of $G$ is $V=\set{0,1}^d$. The terminal set is $T=T_0\cup T_1\cup \set{t^*_0,t^*_1}$, where $t^*_0$ is the all-$0$ vector, $t^*_1$ is the all-$1$ vector, and
\[
T_0=\set{v\mid ||v||_1=\eps d}; \quad\quad T_1=\set{v\mid ||v||_1=(1-\eps) d}.
\]
So $k=|T|=2\cdot\binom{d}{\eps d}+2$. Define $V'=\set{v\mid ||v||_1= d/2}$. The edge set is defined as follows:
\[
E(G)=\bigcup_{v\in V'}\bigg(\bigg\{(v,t)\text{ }\bigg|\text{ } ||v-t||_1=(1/2-\eps)d\bigg\}\cup \bigg\{(v,t^*_0),(v,t^*_1)\bigg\}\bigg),
\]
(we assume that $\eps d$ and $d/2$ are integers for convenience; the proof also holds if we replace $\eps d$ by $\floor{\eps d}$ everywhere)
where the edges $(v,t^*_0),(v,t^*_1)$ have capacity $\binom{d/2}{\eps d}$, or equivalently we can simply add $\binom{d/2}{\eps d}$ parallel edges connecting $v$ to $t^*_0$ and add $\binom{d/2}{\eps d}$ parallel edges connecting $v$ to $t^*_1$. Therefore, every vertex $v\in V'$ has $4\cdot \binom{d/2}{\eps d}$ edges incident to it. We can pair these edges as follows: each edge connecting $v$ to a terminal $t\in T_0$ is paired with a copy of $(v,t^*_1)$, and each edge connecting $v$ to a terminal $t\in T_1$ is paired with a copy of $(v,t^*_0)$. Note that each pair of edges form a shortest path connecting either $t^*_0$ to a terminal in $T_1$ or $t^*_1$ to a terminal in $T_0$, with total $\ell_1$-length $(1-\eps)d$.

Clearly, every edge is incident to some terminal, and so $G$ is indeed a quasi-bipartite graph.

We now show that any solution $f$ to this instance of $\bhc$ with size at most $2^{d/5}$ must have stretch at least $(1+\Omega(\eps))$. Since
$k=2\cdot\binom{d}{\eps d}+2=O(e/\eps)^{\eps d}$, $2^{d/5}=k^{\Omega(1/(\eps\log (1/\eps)))}$.

For each vertex in the image set of $f$, the number of vertices in $V$ that is at distance at most $d/100$ away from it is at most $\sum_{i=0}^{d/100} \binom{d}{i}\le (100e)^{d/100}\le 2^{d/5}$.
As $|V'|=\binom{d}{d/2}=\Omega(2^d/\sqrt{d})>2^{d/2}$, the number of vertices $v\in V'$ with $||f(v)-v||_1\ge d/100$ is at least
$|V'|-2^{d/5}\cdot 2^{d/5}>0.9\cdot |V'|$.

For any such vertex $v$, either there are at least $d/200$ coordinates with value $1$ in $v$ and $0$ in $f(v)$, or there are at least $d/200$ coordinates with value $0$ in $v$ and $1$ in $f(v)$. Assume the former holds (the arguments for the latter is symmetric). 
If we randomly sample a terminal $t$ in $T_0$ that has an edge to $v$, then in expectation, there are at least $\eps d/100$ coordinate with value $1$ in and value $0$ in $f(v)$. Therefore, for such a vertex $v\in V'$, $$\sum_{t\in T: t\sim v}||t-f(v)||_1-\sum_{t\in T: t\sim v}||t-v||_1\ge 2\cdot \frac{\eps d}{100}\cdot 2\cdot \binom{d/2}{\eps d} \ge \Omega(\eps)\cdot \sum_{t\in T: t\sim v}||t-v||_1.$$
As a consequence, the stretch of $f$ is at least
\[
\begin{split}
\frac{\sum_{v\in V}\sum_{t\in T: t\sim v}||t-f(v)||_1}{\sum_{v\in V}\sum_{t\in T: t\sim v}||t-v||_1}
& =1+\frac{\sum_{v\in V}\big(\sum_{t\in T: t\sim v}||t-f(v)||_1-\sum_{t\in T: t\sim v}||t-v||_1\big)}{\sum_{v\in V}\sum_{t\in T: t\sim v}||t-v||_1}\\
& \ge 1+0.9\cdot \Omega(\eps) = 1+\Omega(\eps).
\end{split}
\]
This completes the proof.
\end{proof}

\appendix
\section{Missing Proofs}

\subsection{Proof of \Cref{lem: divide}}
\label{apd: Proof of lem: divide}

Graph $H'$ is constructed as follows. The vertex set is $V(H')=\bigcup_{G'\in \gset}V(H_G')$. The terminal set is the union of $T$ and the set of all vertices of $G$ that appear in more than one graphs in $\gset$. The edge set of $H'$ is simply the union of all edges $E(H_{G'})$. Note that $V(H')$, according to our definition, indeed contains all vertices in any graph $G'\in \gset$, so the $H'$ is well-defined, and it is easy to verify that $|V(H')|\le \sum_{G'\in \gset}|V(H_{G'})|$.

It remains to show that the graph $H'$ constructed above is a $(1+\eps)$ cut sparsifier of $G$ with respect to $T$. Let $(T^1,T^2)$ be a partition of $T$ into two subsets.

On the one hand, consider a min-cut $\hat E$ in $G$ separating $T^1$ from $T^2$, and let $(V^1,V^2)$ be the vertex partition of $V(G)$ formed by the cut $\hat E$, where $T^1\subseteq V^1$ and $T^2\subseteq V^2$. We construct a cut $E'$ in $H'$ with $w(E')\le (1+\eps)\cdot w(\hat E)$ as follows. 
For each $G'\in \gset$, define $T^1_{G'}=T_{G'}\cap V^1$ and $T^2_{G'}=T_{G'}\cap V^2$, and let $E'_{G'}$ be a min-cut in $H'_{G'}$ separating $T^1_{G'}$ from $T^2_{G'}$. 
Then we let $E'=\bigcup_{G'\in \gset}E'_{G'}$.

First we show that $w(E')\le (1+\eps)\cdot w(\hat E)$.
Recall that the edge sets $\set{E(G')\mid G'\in \gset}$ partition $E(G)$. For each $G'\in \gset$, denote $\hat E_{G'}=\hat E\cap E(G')$, so the sets $\set{\hat E_{G'} \mid G'\in \gset}$ partitions $\hat E$ and so $w(\hat E)=\sum_{G'\in \gset}w(\hat E_{G'})$.
Note that for each $G'\in \gset$, the set $\hat E_{G'}$ separates $T^1_{G'}$ from $T^2_{G'}$, and by definition we have $w(E'_{G'})\le (1+\eps) w(\hat E_{G'})$. Altogether, $w(E')\le (1+\eps)\cdot w(\hat E)$.

Second we show that $E'$ indeed separates $T^1$ from $T^2$ in $H'$. For each $G'\in \gset$, denote by $(V^1_{G'},V^2_{G'})$ the partition formed by the cut $E'_{G'}$, where $T^1_{G'}\subseteq V^1_{G'}$ and $T^2_{G'}\subseteq V^2_{G'}$. We let $V^1_{H'}=\bigcup_{G'\in \gset}V^1_{G'}$ and $V^2_{H'}=\bigcup_{G'\in \gset}V^2_{G'}$.
It is easy to verify that (i) $V^1_{H'}$ and $V^2_{H'}$ partition $V(H')$; (ii) $T^1\subseteq V^1_{H'}$; (iii) $T^2\subseteq V^2_{H'}$; and (iv) all edges in $E_{H'}(V^1_{H'},V^2_{H'})$ lies in $E'$, so $E'$ indeed separates $T^1$ from $T^2$ in $H'$.

On the other hand, consider a min-cut $E'$ in $H'$ separating $T^1$ from $T^2$, and let $(V_{H'}^1,V_{H'}^2)$ be the corresponding vertex partition of $V(H')$, where $T^1\subseteq V^1$ and $T^2\subseteq V^2$. 
We construct a cut $\hat E$ in $G$ with $w(\hat E)\le (1+\eps)\cdot w(E')$ as follows. 
For each $G'\in \gset$, define $T^1_{G'}=T_{G'}\cap V^1_{H'}$ and $T^2_{G'}=T_{G'}\cap V^2_{H'}$, and let $\hat E_{G'}$ be a min-cut in $G'$ separating $T^1_{G'}$ from $T^2_{G'}$. 
Then we let $\hat E=\bigcup_{G'\in \gset}\hat E_{G'}$.
The arguments for showing that $w(\hat E)\le (1+\eps)\cdot w(E')$ and that $\hat E$ indeed separates $T^1$ from $T^2$ in $G$ are symmetric.

\subsection{Proof of \Cref{lem: pattern cover}}
\label{apd: Proof of lem: pattern cover}

Let $P$ be the $\Phi$-pattern shortest path connecting $v$ to any vertex on $R$. Let $z$ be the $R$-endpoint of $P$, and let $\ell$ be the total length of $P$.
Vertex $z$ separates $R$ into two subpaths, that we call the \emph{forward} subpath from $z$ to the one endpoint of $R$ and the \emph{backward} subpath from $z$ to the other endpoint of $R$.
We now iteratively computes a set $C_F$ of vertices in the forward subpath as follows.
Initially set $C_F=\emptyset$ and a variable vertex $u=z$. In each iteration, we find the first vertex $u'$ lying in the forward direction from $u$, such that 
$\dist^{\Phi}(v,u)+\dist(u,u')> (1+\eps)\cdot \dist^{\Phi}(v,u')$. We add $u$ to $C_F$ and update $u\leftarrow u'$. We continue until we cannot find any such $u'$ in the forward path and obtain the resulting $C_F$.
We then computes a set $C_B$ of vertices in the back subpath in the same way, and eventually returns $C(v,R,\Phi)=C_F\cup C_B$ as the $\Phi$-respecting $\eps$-cover.
Clearly, by definition we indeed get a $\Phi$-respecting $\eps$-cover of $v$ in $R$. The next claim shows that $|C(v,R,\Phi)|=O(1/\eps)$.

\begin{claim}
The number of iterations  is $O(1/\eps)$.
\end{claim}
\begin{proof}
Let $u_1,u_2,\ldots$ be the sequence of vertices we added into $C_F$. On the one hand, by triangle inequality, $-\dist^{\Phi}(v,z)\le \dist^{\Phi}(v,u_i)-\dist(z,u_i)\le \dist^{\Phi}(v,z)$  for all $i\ge 1$. On the other hand, for every $i$, 
$$\dist^{\Phi}(v,u_{i+1})<\frac{1}{1+\eps}\cdot \bigg(\dist^{\Phi}(v,u_i)+\dist(u_i,u_{i+1})\bigg)\le \dist^{\Phi}(v,u_i)+\dist(u_i,u_{i+1})-\frac{\eps}{2}\cdot \dist^{\Phi}(v,z),$$
which means that
\[
\bigg(\dist^{\Phi}(v,u_{i+1})-\dist(z,u_{i+1})\bigg)\le \bigg(\dist^{\Phi}(v,u_i)-\dist(z,u_i)\bigg)-\frac{\eps}{2}\cdot \dist^{\Phi}(v,z).
\]
Therefore, we can find at most $4/\eps$ such $u_i$. This completes the proof of the claim.
\end{proof}


\bibliographystyle{alpha}
\bibliography{REF}

\end{document}